\documentclass[Journal]{IEEEtran}
\ifCLASSINFOpdf
  \else
  \fi

\usepackage{subcaption}
\usepackage{amsfonts}       
\usepackage{nicefrac}       
\usepackage{microtype}    
\usepackage{amsfonts}   
\usepackage{nicefrac}   
\usepackage{times}
\usepackage{adjustbox}
\usepackage{tabularx}
\usepackage{amsmath}
\usepackage{multirow}
\usepackage{comment}
\usepackage{amsthm}
\usepackage{blindtext}
\usepackage{float}
\usepackage{enumitem}
\usepackage{multirow}
\usepackage{graphicx}
\usepackage{bm}
\usepackage{algorithm, algorithmic}
\usepackage{amsfonts}
\usepackage{amsthm}
\usepackage{color}

\begin{document}
\newtheorem{emp}{Example}
\newtheorem{cn}{Conjecture}
\newtheorem{lm}{Lemma}
\newtheorem{thm}{Theorem}
\newtheorem{cor}{Corollary}
\newtheorem{df}{Definition}
\newtheorem{pf}{Proof}
\newtheorem{cond}{Condition}
\newtheorem{alg}{Algorithm}
\newtheorem{rmk}{Remark}

\title{On the Foundation of Sparse Sensing (Part I): Necessary and Sufficient Sampling Theory and Robust Remaindering Problem}

\hyphenation{IEEE Transactions on Signal Processing}

\author{Hanshen Xiao, Yaowen Zhang, and Guoqiang Xiao

\thanks{Hanshen Xiao is with CSAIL and the EECS Department, MIT, Cambridge, USA. E-mail: hsxiao@mit.edu.}
\thanks{Yaowen Zhang and Guoqiang Xiao are with the College of Computer and Information Science, Southwest University, Chongqing, China. E-mail: gqxiao@swu.edu.cn}
}

\maketitle

\begin{abstract}
\noindent In the first part of the series papers, we set out to answer the following question: given specific restrictions on a set of samplers, what kind of signal can be uniquely represented by the corresponding samples attained, as the foundation of {\em sparse sensing}. It is different from {\em compressed sensing}, which exploits the sparse representation of a signal to reduce sample complexity (compressed sampling or acquisition). We use {\em sparse sensing} to denote a board concept of methods whose main focus is to improve the efficiency and cost of sampling implementation itself. The “sparse” here is referred to sampling at a low temporal  or spatial rate (sparsity constrained sampling or acquisition), which in practice models cheaper hardware such as lower power, less memory and throughput. 

We take frequency and direction of arrival (DoA) estimation as concrete examples and give the necessary and sufficient requirements of the sampling strategy. Interestingly, we prove that these problems can be reduced to some (multiple) remainder model. As a straightforward corollary, we supplement and complete the theory of {\em co-prime sampling}, which receives considerable attention over last decade. 

On the other hand, we advance the understanding of the robust multiple remainder problem, which models the case when sampling with noise. A sharpened tradeoff between the parameter dynamic range and the error bound is derived. We prove that, for $N$-frequency estimation in either complex or real waveforms, once the least common multiple (lcm) of the sampling rates selected is sufficiently large, one may approach an error tolerance bound independent of $N$.
\end{abstract}


\section{introduction}
\noindent In 1928, Nyquist pointed out that for an arbitrary band-limited signal, once a system uniformly samples it at a rate that exceeds the signal’s highest frequency by at least a factor of two, termed as the Nyquist rate, the discrete sample sequence acquired can perfectly recover the original signal \cite{Nyquist}. In 1948, Shannon rigorously proved Nyquist's claim and pointed out that Nyquist rate is indeed also the necessary condition that a uniform sampling can uniquely recover the original band-limited signal in general \cite{Shannon}. The Nyquist-Shannon Sampling Theorem serves as the foundation of modern signal processing, which essentially characterizes the relationship between discrete samples and a continuous signal. As for the non-uniform sampling, Landau in 1967 gave a generic conclusion, known as the {\em Landau necessary condition}, that the average sampling rate must be twice the occupied bandwidth \cite{Landau}. However, those strong necessary conditions are not the end of the study on sampling theory:  Nyquist rate can be unnecessary given {\em more assumptions on signals} and undersampling becomes possible provided stronger prior knowledge. 

One representative is compressed sensing \cite{cs1, cs2, cs3, cs4}, rooted in an under-determined linear system, $y=\Phi x$, with $x \in \mathbb{R}^d$ and $\Phi \in \mathbb{R}^{d' \times d}$. In general, finding the unique solution $x$ is impossible when the number of measurements $d'$ is less than the dimensionality $d$. Nonetheless, with sparsity restriction, say $\|x\|_{0} \ll d$, and proper randomness in the encoding matrix $\Phi$, a sparse $x$ can be reconstructed with high probability by solving an $L_1$ convex programming $\min_{z \in \mathbb{R}^d} \|y - \Phi z\|_{1}$ \cite{Decoding}. The above statements provide a generic framework for signal reconstruction if a signal can be sparsely represented in some domain. Theoretically, the number of measurements can be almost only dependent on the sparsity parameter $\|x\|_{0}$ for a high-probability reconstruction. 

However, the implementation of {\em compressed sampling} is not always trivial in practice, which is especially true when handling the analog signal, i.e., how to transform the signal to a domain, where it can be represented by a sparse vector $x$ and acquire samples through random linear measurements. In general, before analyzing, communication or even storage of an analog signal, one has to first convert it into a bit sequence digitally. Different from simply a data processing, the digital representation of a signal is restricted by the power, memory and other hardware constraints of the sampler applied. Though compressed sensing presents a theoretical framework which allows fewer measurements to reconstruct a sparse vector, the cost of each measurement may not be always low. \footnote{We have to stress that the point we want to make here is that {\em not all} (sparse) signal processing tasks perfectly fit the compressed sensing framework due to the sampler restriction. However, in some scenarios, for example in Magnetic Resonance Imaging \cite{MRI2007, MRI2008, MRI2010} and some optical imaging systems \cite{optical}, the signal of interest has already been represented by a sparse vector and sampling operator itself can be modeled by a linear sensing, and thus compressed sensing techniques definitely strengthen the sampling efficiency.}

This raises a fundamental question that, given specific restrictions, what kind of signals can be represented by the corresponding sample sequence. Low energy-cost hardware with loose architecture requirement is always desired in many application scenarios. Recently, the series works \cite{distortion2015, distortion2018, distortion2018a}, by Kipnis, Eldar and Goldsmith, give a complete characterization of the fundamental trade-off among sampling rate, compression bitrate (quantification rate) and reconstruction error under the analog-to-digital framework, which generalizes the classic sampling theorem to more practical situations. 

On the other hand, to enable undersampling beyond the sampling theorem restriction, roughly speaking, there are two main relaxation directions. One is to make proper assumptions on the signal form, say standard complex waveforms, where reconstructing the signal is instance-based and equivalent to determining the unknown parameters, such as frequency and phase. The other is to use multiple undersampling samplers with a finer grained coordination. On the algorithmic level, there are two successful examples emerged during last two decades, the Chinese Remainder Theorem (CRT) based parameter estimation, short for CRT method in the following, pioneering by Xia \cite{crt1999, crt2000}, and the co-prime sampling (array), pioneering by Pal and Vaidyanathan \cite{co-prime2010, co-prime2011}. 

\noindent \textbf{CRT Based Reconstruction:} CRT method is spectrum based. A complex waveform $x(t) = e^{j2\pi ft}$, with an unknown  frequency $f$, is sampled uniformly at a rate $m$, $m<f$, and Fourier transform is applied to the discrete sequence. It is well known that undersampling will result in a spectrum aliasing: the location of the spectrum peak is indeed the $f$ modulo $m$, denoted by $\langle f \rangle_{m}$ in the following. Said another way, when $m<f$, the residue, $\langle f \rangle_{m} = f - m \lfloor  f/m \rfloor$, is of ambiguity to true $f$, unless the folding number $\lfloor  f/m \rfloor$ can be determined. On the other hand, CRT characterizes the relationship between a number and its residues modulo multiple smaller numbers, termed as moduli. Mathematically, CRT states that given $L$ moduli, the set of residues can uniquely determine a positive number $f$ once the least common multiple (lcm) of the moduli is larger than $f$. Therefore, with a careful selection of sampling rates, say pairwise co-prime, a frequency $f$ can be reconstructed via exponentially small sampling rates $O(f^{\frac{1}{L}})$. This idea is further generalized to other parameter estimations, such as phase unwrapping \cite{phase} with applications in synthetic aperture radar (SAR) \cite{sar2004, sar2017}. 

However, in practice, CRT method encounters two main challenges. One is {\em robustness}, where the residue system is non-weighted and conventional CRT is very sensitive to small error. The follow-up works \cite{rcrt2010, rcrt2014, ring-CRT} addressed the problem by adding redundancy to the moduli (sampling rate) by setting all moduli to share a sufficiently large common divisor. The other challenge is the {\em correspondence relation ambiguity} between multiple numbers and their residues. Though a single objective can be elegantly estimated in the above model, things becomes trickier when handling multiple objectives simultaneously. In general, for a signal $x(t) = \sum_{i=1}^N e^{j2\pi f_it}$ of multiple frequencies $\{f_1, f_2, ... ,f_{N}\}$, the correspondence between the spectrum peaks and the frequencies is indistinguishable without any further assumptions. Therefore, this problem cannot be simplified to $N$ independent single-frequency estimations. One strategy is to sacrifice the encoding dynamic range, i.e., the range of $f_i$ can be uniquely determined, to produce sufficient redundancy to enable decoding. Such tradeoff between the dynamic range and $N$ is studied in \cite{sharpened},  \cite{error-dynmaic-tradeoff, error-dynmaic-tradeoff-2}, whereas the tight bound is only known when $N=2$ \cite{twonumcrt}. To tackle both robustness and correspondence ambiguity simultaneously, the polynomial-time (statistical) robust multi-objective estimation is only known recently in \cite{multinumcrt2016, multinumrcrt, tvt2019, sp2021}.

\noindent \textbf{Co-prime Sampling (Array):} Different from CRT method, which handles the spectrum directly, co-prime sampling estimates the auto-correlation on the temporal domain. Technically, the foundation of co-prime sampling is Bezout Theorem, another fundamental number theory theorem, which states that for arbitrary two numbers $P$ and $Q$, there exist $\alpha$ and $\beta$ such that 
$$P \beta - Q \alpha = gcd(P,Q),$$
where $gcd(P,Q)$ denotes the greatest common divisor of $P$ and $Q$. Especially, if $P$ and $Q$ are co-prime, i.e., $gcd(P,Q) = 1$, as a straightforward corollary, there exists $\beta_1$ and $\alpha_1$ such that $$P \beta_1 - Q \alpha_1 = 1,$$ and multiplying $l$ on both sides, we have 
$$ P \cdot (l \beta_1) - Q \cdot (l \alpha_1) = l,$$
for arbitrary $l$.

With the above preparation, now we can give a high-level picture of co-prime sampling. Assume that the time interval corresponding to Nyquist rate is $T$ and two samplers with undersampling rates $P$ and $Q$, respectively, produce two discrete sample sequences $\{x[PTn], n \in \mathcal{N}\}$ and $\{x[QTn], n \in \mathcal{N}\}$ of a continuous signal $x(t)$. From Bezout theorem, once $P$ and $Q$ are co-prime, there always exist $\alpha_l$ and $\beta_l$ such that 
$$ P \cdot \beta_l - Q \cdot  \alpha_l = l,$$ and correspondingly the pair $(x[P\beta_l T], x[Q\alpha_l T] )$ renders an estimation of the auto-correlation $\mathbb{E}_{n} \big\{ x[nT] \cdot x^*[(n-l)T] \big\}$ at lag $l$. Since auto-correlation and signal spectrum are a pair of Fourier transform, the frequencies contained in a signal can be estimated from the signal spectrum.

Switching to the spatial case, where $T$ is replaced by the half of wavelength, the spatial sampling becomes to determine the physical location of sensors and such construction is termed as co-prime array \cite{co-prime2011}. 

Though both CRT method and co-prime sampling are rooted in number theory, or more precisely, the Diophantine equation theory, they were separately studied over last decade. As far as we know, our paper is the first work to connect them together and take them as the foundation of the study on sparse sensing. In short, our conclusion is somewhat surprising: \textbf{The underlying model in CRT method is closer to characterize the necessary and  sufficient condition that a signal can be uniquely represented in sparse sensing, while co-prime sampling indeed gives a more efficient “high-probability” reconstruction.} \footnote{The high probability here is either with respect to natural distribution of the parameter to be estimated, or a random undersampling rate selection for any given parameters to be estimated, which will be specified later.} 

We start with these two representatives of sparse sensing to give some intuition that the two fundamental questions we set out to address. 
\begin{enumerate}
    \item First, theoretically, how “sparse”, regarding both the temporal and spatial sampling rate, can we expect to be? Or equivalently, given the sparsity restriction on the sampling, what kind of signal can be uniquely represented by the attained samples.
    \item Second, under sparsity restriction, algorithmically, how “efficient”, regarding both the freedom degree of parameters to be estimated and the time complexity with associated performance, can the reconstruction scheme be? Also, can the efficiency be significantly improved by resorting to a small failure probability, relaxed from a deterministic reconstruction?   
\end{enumerate}

This paper mainly focuses on the first “encoding limitation” question to give a systematical study on the necessary and sufficient conditions of sparse sampling theory. Concrete sampling algorithms (“decoding”) will be presented in the second paper, where interestingly we show given sufficiently many samplers/senors, it could be sparse enough, i.e., arbitrarily low sampling frequency in temporal domain and arbitrarily large spacing between sensors, with a linear sampling time. 

\noindent \textbf{Organization and Contributions}: We summarize the organization of rest contents as follows. In Section \ref{sec:n-s-c}, we take frequency and DoA estimation as two concrete examples and give the necessary and sufficient conditions of the signal that can be represented given sampling sparsity restrictions. In Section \ref{sec:coprime_theory}, we give a complete analysis of co-prime sampling in Theorem \ref{co-prime complete}, which indicates that
co-prime sampling is indeed a high-probability reconstruction. In Section \ref{sec:comple_fre}, we consider the more general case where sampling with noise and we study this problem from a view of error-correcting codes. We provide a characterization of encoding dynamic range and error tolerance of complex waveforms in Theorem \ref{thm:complex_fre}. In Section \ref{sec:real_fre}, we provide further generalization to the scenario of real waveforms. We conclude in Section \ref{conclusion}.

\section{Necessary and Sufficient Conditions}
\label{sec:n-s-c}
\noindent In general, sampling can be viewed as a encoding procedure, where the objective continuous signal is transformed to a discrete sample sequence. Therefore, the necessary and sufficient conditions of sampling are equivalent to determining if there is a bijection between the objective signal and the encoded sample sequence. With the specification on the assumptions of signals and restrictions on the samplers, such conditions are deterministic, at least theoretically in the noiseless scenario. In the following, we take frequency and DoA estimation as two concrete examples in the temporal and spatial domain, respectively.

\subsection{Frequency and Phase Estimation}
\label{fre-est-sec}
\noindent Without loss of generality, we consider a complex waveform 
\begin{equation}
\label{original-sig}
  x(t) = \sum_{i=1}^N A_i e^{j(2\pi f_i t + \theta_i)},  
\end{equation}
where $A_i$ is the amplitude, $f_i$ is the frequency and $\theta_i$ is the phase of the $i$-th source. The continuous signal $x(t)$ is sampled by a group of samples with sampling rates $\{m_1, m_2, ... ,m_L\}$, respectively. Accordingly, $L$ sequences are obtained in a form, 
\begin{equation}
\label{sampled-seq}
   x_l[n] = \sum_{i=1}^N A_i e^{j\big(2\pi (f_i/m_l) n + \theta^l_i\big)},
\end{equation}
where $\theta^l_i$ is the phase of the the $i$-th source at the $l$-th sampler. Since Fourier Transform (FT) is an invertible transformation, applying FT to (\ref{sampled-seq}), we have, 
\begin{equation}
\label{FT-form} 
   \mathcal{F}_l(k) = FT(x_l[n]) = \sum_{i=1}^n A_i \delta(k - \langle f_i \rangle_{m_l}) e^{j\theta^l_i}, 
\end{equation}
where $\delta(k)$ denotes the indicator, which equals $0$ iff $k=0$. Without any prior assumptions on the rest parameters $\theta^l_i$ and $A_i$, the sequence $x_l[n]$ at the $l$-th sampler is uniquely characterized by an {\em unordered} set $\mathcal{R}_l = \{ r_{il} = \langle f_i \rangle_{m_l}, i=1,2,...,N\}$. The necessary and  sufficient conditions that the unknown frequency set $\bm{f} = \{f_1, f_2, ... ,f_n \}$ can be uniquely reconstructed with the $L$ sequences together are that, $N$ frequencies $\bm{f} \in [0,D)^N$, where $[0,D)$ represents the dynamic range, are bijective to the encoded $L$ unordered residue sets $\mathcal{R}_{[1:L]}$. It is worth mentioning that the above statement relies on the fact that no additional prior knowledge is assumed. If $A_i$ or $\theta^l_i$ can be distinguished, $\mathcal{R}_l$ becomes an “ordered” set, where the correspondence between each $r_{il}$ and $f_i$ is determined. Given the correspondence, the necessary and  sufficient conditions become $D < lcm(m_1, m_2, ... ,m_L)$, characterized by regular CRT. We summarize the conclusion as the following theorem.

\begin{thm}
\label{nsc-fre}
Without additional prior assumptions, the necessary and  sufficient conditions that $N$ frequencies $\bm{f}= \{f_1, f_2, ... ,f_N\}$ within $[0,D)^N$ defined in (\ref{original-sig}) can be reconstructed given a set of samplers with rates $\mathcal{M}=\{m_1, m_2, ... ,m_L\}$ are that there exists bijective correspondence between $\bm{f}$ and $\mathcal{R}_{[1:L]}$, where $\mathcal{R}_{l} = \{ r_{il} = \langle f_i \rangle_{m_l}, i=1,2,...,N \}$ is an unordered residue set. Especially, if $A_i$ or $\theta^l_i$ are assumed to be distinct and distinguishable, the necessary and  sufficient conditions become $D < lcm(m_1, m_2, ... ,m_L)$.
\end{thm}

Analogous to the frequency estimation, it is not hard to see that the necessary and sufficient conditions of the phase-difference based distance measurement share a similar form, where $\bm{f}$ will correspond to the distance(s) between the target and the sensors to be estimated. We will specify the dynamic range $D$ in Section \ref{sec:comple_fre} and generalize the above analysis to real waveforms in Section \ref{sec:real_fre}.

\subsection{DoA Estimation}
\noindent With a similar reasoning, we can also formulate the DoA estimation as an encoding and decoding problem. Analogous to undersampling a signal in time domain, spatial undersampling indeed also exists when the distance between adjacent array sensors is larger than a half of wavelength. In general, assume there are $N$ independent signals $\{s_i(t), i=1,2,...,N \}$, whose DoAs are $\bm{\theta}_{[1:N]} = \{\theta_1, \theta_2, ... , \theta_N\}$, respectively. Consider an linear array of $L$ sensors, which are placed at $\{p_1, p_2, ... ,p_L\}$, where without loss of generality, $p_i \geq 0$ and $p_1$ is set to be $0$ as the reference. Thus, the sample sequence at the $l$-th sensor can be written as 
\begin{equation}
\label{doa-formula}
    x_l[n] = \sum_{i=1}^N e^{-j\frac{2\pi p_l}{\lambda}\sin(\theta_i)} s_i[n]. 
\end{equation}
Therefore, let $\bm{x}[n]=(x_1[n], x_2[n], ... , x_L[n])^T$, and (\ref{doa-formula}) can be rewritten in a matrix form,
\begin{equation}
\label{doa-matrix-formula}
 \bm{x}[n] = \bm{A} \bm{s}[n],
\end{equation}
where $\bm{A}=(\bm{a}(\theta_1), \bm{a}(\theta_2),\dots,\bm{a}(\theta_N))$ and $\bm{s}[n]=(s_1[n],\dots,s_N[n])^T$. $\bm{A}$ is usually called {\em array manifold matrix}, where $\bm{a}(\theta_i)= [1, e^{-j\frac{2\pi p_2}{\lambda} \sin(\theta_i)}, ... ,e^{-j\frac{2\pi p_L}{\lambda} \sin(\theta_i)}]^T $, usually called the {\em steering vector}. Here, $\lambda$ represents the wavelength. Moreover, if each sensor takes $K$ snapshots for $n=1,2,...,K$, the samples set $\bm{x}[1:K] = \bm{A}\bm{S}$, where $\bm{S}=(\bm{s}[1], \bm{s}[2], ... ,\bm{s}[K])$. We assume $K>N$ in the following to guarantee sufficient samples. Without any restriction, $\bm{\theta}_{[1:N]}$ are distributed within $[-\pi/2, \pi/2)^N$. Generally, the distance between any adjacent sensors is set to be $\lambda/2$ for most uniform arrays. When the distance between adjacent sensors is larger than $\lambda/2$, there is an ambiguous range for the $l$-th sensor \cite{benesty2017fundamentals} $$\Psi_{\theta_i} = [-\frac{\pi}{2}, -\arcsin(\frac{\lambda}{p_l})] \cup [\arcsin(\frac{\lambda}{p_l}), \frac{\pi}{2})$$ regarding $\theta_i$. In the following, we present a generic characterization on the condition that $N$ DoAs can be uniquely encoded given an arbitrary array.  

\begin{thm}
\label{nsc-doa}
When $\bm{S}$ is a full-rank matrix, the necessary and sufficient conditions that $N$ DoAs, $\bm{\theta}_{[1:N]}$, can be uniquely reconstructed by sample sequences from an array of $L$ sensors with locations specified as $\{p_1=0, p_2, ... ,p_L\}$ are that $\mathcal{C} \geq 2$, where 
\begin{equation}
   \mathcal{C} = \min_{c_2, ... ,c_L \in \mathbb{Z}^+} \frac{c_2\lambda}{p_2} = \frac{c_3\lambda}{p_3} = ... = \frac{c_L\lambda}{p_L},
\end{equation}
is the least common multiple (lcm) of the numbers $\{ \frac{\lambda}{p_2}, \frac{\lambda}{p_3}, ... , \frac{\lambda}{p_L}\}.$
\end{thm}

\begin{proof}
Since $\bm{S}$ is assumed to be of full-rank, the encoding procedure between $\bm{A}$ and $K$ snapshots $\bm{x}[1:K]$ is invertible. Thus, it is equivalent to considering the condition that the mapping between $\bm{\theta}_{[1:N]}$ and $\bm{A}$ is bijection, which means for arbitrary $\theta \not = \theta'$, $\bm{a}(\theta) \not = \bm{a}(\theta')$. On the other hand, if $\bm{a}(\theta) = \bm{a}(\theta')$, it is equivalent to that, there exist $k_l \in \mathbb{Z}$, $l=1,2,...,L$ and 
\begin{equation}
\label{doa1-ambiguity}
 2\pi \frac{p_l}{\lambda} \sin(\theta_i) = 2\pi (\frac{p_l}{\lambda} \sin(\theta'_i)+ k_l).   
\end{equation}
Rewriting (\ref{doa1-ambiguity}) into a modulo form, we have 
\begin{equation}
\label{doa-ambiguity}
\sin(\theta_i) \equiv \sin(\theta'_i) \mod \frac{\lambda}{p_l}, l=1,2,...,L.    
\end{equation}
Here, we use the generalized modulo operation where for any two real numbers, $\langle a \rangle_b = a - \lfloor \frac{a}{b} \rfloor b$. Thus, a bijection is equivalent to that arbitrary $\sin(\theta_i) \in [-1,1)$ has a unique residue representation modulo $\frac{\lambda}{p_l}, l=2,3,...,L.$ Hence, the rest of the proof is indeed to describe a generalized CRT with real moduli: Indeed, CRT for real numbers is analogous to the integer case, where the largest representability is still determined by a generalized least common multiple of $\{\frac{\lambda}{p_l}, l=2,3,...,L\}$, defined by 
$$ \mathcal{C} = \min_{c_2, ... ,c_L \in \mathbb{Z}^+} \frac{c_2\lambda}{p_2} = \frac{c_3\lambda}{p_3} = ... = \frac{c_L\lambda}{p_L}. $$ The proof is straightforward. On one hand, it is trivial to verify that the residues of $a=0$ and $b=\mathcal{C}$ modulo $\{\frac{\lambda}{p_l}, l=2,3,...,L\}$ are identical. On the other hand, for any two real numbers $a, b \in [0,\mathcal{C})$, if they share exactly the same residues, which implies $$ a-b \equiv 0 \mod \frac{\lambda}{p_l},$$ for $l=2,3,...,L$, a contradiction to the minimization assumption of $\mathcal{C}$. Finally, it is noted that $\sin(\theta) \in [-1,1)$, therefore we need $\mathcal{C}$ to be at least $2$ and the claim follows. 
\end{proof}

Analogous to the DoA estimation, in the case of phase unwrapping (Doppler shift) in multi-frequency antenna array synthetic aperture
radar (SAR), the condition is also similar to Theorem \ref{nsc-doa}.

Throughout the two examples, we have shown that the fundamental problems to characterize the frequency and DoA estimation under sparse sampling are closely related to the remainder model. This is not surprising due to the cycle property of waveform, where sampling itself can be viewed as somewhat modulo operation on the signal. Though the necessary and  sufficient conditions describe a preliminary picture that when a signal can be perfectly reconstructed given samplers with specific sparsity restriction, there are two important generalizations worth further exploring: 
\begin{enumerate}
    \item How could restrictions in Theorem \ref{nsc-fre} and \ref{nsc-doa} be further relaxed resorting to some failure probability? Beyond perfect reconstruction on all signals, compressed sensing provides a successful example to produce high probability recovery but much sharpened sample complexity via random sampling.  
    
    \item Theorem \ref{nsc-fre} and \ref{nsc-doa} only characterize the noiseless case while the noise and quantification error are almost inevitable in practice. With further robustness consideration, how should Theorem \ref{nsc-fre} and \ref{nsc-doa} be modified.  
\end{enumerate}
We address the above two questions in the following.

\section{A Complete Analysis of Co-prime Sampling}
\label{sec:coprime_theory}
\noindent In the Introduction, we briefly mentioned the key idea of co-prime sampling  \cite{co-prime2010, coprime-conf}. The elegant construction by selecting two co-prime under-sampling rates, $P$ and $Q$, gives a unified approach to estimate the auto-correlation of the sample sequence attained by Nyqiust rate. Due to Bezout theorem, one can properly select infinitely many pairs from two undersampled sequences such that their sampling-time difference to be any multiple of the time interval sampled by Nyquist rate. Once one can exactly recover the auto-correlation of the sequence sampled at Nyquist rate, according to Wiener–Khinchin theorem, the power spectrum can be reconstructed by applying Fourier transform on the auto-correlation.

Co-prime sampling theory seems to violate the necessary and sufficient condition described in Theorem \ref{nsc-fre}, since $P$ and $Q$ can be arbitrarily large, i.e., one can use two arbitrarily small sampling rates to recover the spectrum of a signal with a arbitrarily large bandwidth. As shown below, we will rigorously prove co-prime sampling is indeed an efficient “high-probability” reconstruction but a failure could happen.   

Assuming that two samplers whose sampling rates are $m_1= 1/{PT}$ and $m_2 = 1/{QT}$, respectively, where the Nyquist rate is $1/{T}$. With the same setup described in Section \ref{fre-est-sec}, $x(t) = \sum_{i=1}^N e^{j 2\pi f_i t}$, where we normalize the phase and amplitude without loss of generality, is sampled and two sequences are attained as 
\begin{equation}
\label{coprime-seq}
    x_1[n_1] = \sum_{i=1}^N e^{j 2\pi f_i TPn_1}, ~~ x_2[n_2] = \sum_{i=1}^N e^{j 2\pi f_i TQn_2}. 
\end{equation}
Since $P$ and $Q$ are co-prime integers, the two sequences indeed form a cycle of time $PQT$. In the $k$-th cycle, i.e., we consider the sampling time at $n_1 = Qk+r_1$ and $n_2 = Pk+r_2$, and denote the sample vectors by $\bm{x}^1_k = (x_1[Qk+1], x_1[Qk+2], ... ,x_1[Q(k+1)])$ and $\bm{x}^2_k = (x_2[Pk+1], x_2[Pk+2], ... ,x_2[P(k+1)])$.  Following \cite{co-prime2010, coprime-conf}, let $\tilde{\bm{x}}_k = (\bm{x}^1_k, \bm{x}^2_k)$, be the Cartesian product of the two vectors $\bm{x}^1_k$ and $\bm{x}^2_k$, and consider $\bm{R} = \mathbb{E}_l \big \{(\tilde{\bm{x}}^l)^T \cdot (\tilde{\bm{x}}^l)^* \big\}$, which ideally is expected to be the auto-correlation estimation. Here, $\bm{a}^T$ and $\bm{a}^*$ denotes the transpose and conjugate of a vector $\bm{a}$, respectively.  

When $N=1$, the co-prime sampling theory is indeed a special case of Theorem \ref{nsc-fre}. It is noted that \begin{equation}
    lcm(\frac{1}{PT}, \frac{1}{QT}) = \big(\min_{c_1, c_2 \in \mathbb{Z}^+} \frac{c_1}{PT} =  \frac{c_2}{QT} \big) = \frac{1}{T},
\end{equation}
where the quality holds when $c_1=P$ and $c_2=Q$ due to the co-prime assumption on $P$ and $Q$. 

However, when $N>1$, co-prime sampling comes with a failure, where the estimation $\bm{R}$ could be biased. Consider $\mathbb{E}_k\{ x_1[Qk+u] \cdot x_2[Pk+v]\}$, which is took as an estimation of auto-correlation at lag $(uP-vQ)$, i.e., $\mathbb{E}_{n} \big\{x[nT] \cdot x^*[(n-(uP-vQ))T] \big\}$. With some simple calculation,
\begin{equation}
\label{difference_1}
\begin{aligned}
    \mathbb{E}_k & \{ x_1[Qk+u] \cdot x^*_1[Pk+v]\}  \\
    & =  \mathbb{E}_{n} \{x[nT] \cdot x^*[(n-(uP-vQ))T]\} \\
    & + \underline{\sum_{i \not= l}\mathbb{E}_{k} \big[e^{j 2\pi f_i P (Qk+u)T} \cdot e^{-j 2\pi f_l Q (Pk+v)T}  \big]}
\end{aligned}
\end{equation}
Similarly, if the estimation is from the cross product from one sample sequence, say $x_1[\cdot]$, where we expect $\mathbb{E}_k \big\{ x_1[kQ+u] \cdot x_1[kQ+v] \big\}$ equals the auto-correlation at lag $(u-v)P$, i.e., $\mathbb{E}_{n} \big\{x[nT] \cdot x^*[(n-(u-v)P)T] \big\}$, it becomes 
\begin{equation}
\label{difference_2}
\begin{aligned}
    \mathbb{E}_k & \{ x_1[Qk+u] \cdot x^*_1[Pk+v]\}  \\
    & =  \mathbb{E}_{n} \{x[nT] \cdot x^*[(n-(u-v)P)T]\} \\
    & + \underline{\sum_{i \not= l}\mathbb{E}_{k} \big[e^{j 2\pi f_i P (Qk+u)T} \cdot e^{-j 2\pi f_l P (Qk+v)T}  \big]}
\end{aligned}
\end{equation}
Thus, co-prime sampling forms an unbiased estimation if and only if the underlying terms in (\ref{difference_1}) and (\ref{difference_2}) should equal 0. With a bit more calculation, a sufficient condition is that, 
\begin{equation}
 \mathbb{E}_{k} e^{j2\pi k \frac{PQ(f_i-f_l)}{f_s}} = 0 
\end{equation}
which requires that $PQ (f_i - f_l )T$ cannot be an integer. 

\begin{thm}
\label{co-prime complete}
When $N>2$, the co-prime sampler defined in (\ref{coprime-seq}) can perfectly reconstruct the $N$ frequencies $\{f_1, f_2, ... ,f_N\}$, when for any $i \not = j \in [1:N]$, $PQ(f_i - f_j)$ cannot be divided by the Nyquist rate $1/{T}$. \footnote{Another takeaway from the above analysis is that, in general by applying the co-prime sampling idea on generic continuous signal, the assumption that the components are independent is not sufficient since after sampling (processing), the independence may not hold.}
\end{thm}

Theorem \ref{co-prime complete} characterizes a failure set, within which the signal cannot be uniquely reconstructed. However, it is worth noting that such set is very sparse (indeed has zero measure) across the whole domain $[0,1/T]^{N}$. Compared to the best known dynamic range in CRT decoding method \cite{TSP2018}, which states that given $L$ samplers with co-prime rates $\{m_1, m_2, ... , m_L\}$, one can uniquely recover $\bm{f} \in [0,D]^N$ for $D = \Omega\big( (\prod_{l=1}^L m_l)^{1/N} \big)$. Theorem \ref{co-prime complete} is inspiring and gives a positive example that, resorting to a negligible failure probability (when $\bm{f}$ is assumed to be uniformly distributed, or a random selection of $P$ and $Q$), it is possible that one can use much sharpened undersampling rate to almost reconstruct the signal, compared to the deterministic decoding bound. 

\section{Robust Reconstruction for Frequency Estimation of Complex Waveform}
\label{sec:comple_fre}
\noindent Throughout the previous sections, we have derived the necessary and sufficient reconstruction condition in the {\em noiseless} case. In the rest of this paper, we set out to advance the understanding of robust reconstruction and stick to the more complicated (multiple) frequency estimation model. We adopt the noise model considered in  \cite{phase, rcrt2010, rcrt2014, twonumrcrt, multinumrcrt} and study the robustness from an error-correcting decoding viewpoint.

As described in Theorem \ref{nsc-fre}, under $L$ sampling rates $m_{[1:L]}$, the frequency set $\bm{f}=(f_1, f_2, ... , f_N)$ to be estimated is essentially encoded into $L$ unordered residue sets $\mathcal{R}_{[1:L]}$, where $\mathcal{R}_{l}=\{ \langle f_i \rangle_{m_l}, i=1,2,...,N \}$. To formally study the robust decoding condition, we introduce the {\em robust remaindering problem} as follows.

\begin{df}[Robust Remaindering Problem]
Given $m_{[1:L]}$ and $N$, what is the relationship between the dynamic range $D$ and the maximal error bound $\delta$ such that for arbitrary $\bm{X}=(X_1, X_2, ... , X_N) \in [0,D)^N$, provided $L$ noisy unordered residue sets $\mathcal{R}_{[1:L]}$, where $\mathcal{R}_{l}=\{ \langle X_i \rangle_{m_l}+\Delta_{il}, i=1,2,...,N \}$ for arbitrary noise $|\Delta_{il}| < \delta$, $\bm{X}$ can be robustly estimated by $\hat{\bm{X}}=\{\hat{X}_1, \hat{X}_2, ... , \hat{X}_N \}$ and the estimation error $|X_i - \hat{X}_i |  = O(\delta)$ for any $i$. 
\end{df}

Before we proceed to study the tradeoff between $D$ and $\delta$ in general, we first provide the following results on the sampling rate selection. We first move our attention to a single number case, i.e., $N=1$. In the following, we prove the optimal modulus selection in terms of the error tolerance capacity. 

\subsection{Sampling Rate Selection}
\noindent In \cite{TSP2018}, it is shown that when $N=1$, the maximal error bound is,
\begin{equation}
\delta = \min \max_{ \min m_l \leq X < D} \frac{\Vert{\textbf{x}}\Vert_{\infty}}{4}
\end{equation}
Here, $\textbf{x}$ denotes the residue representation of a number $X$ and $\Vert{\textbf{x}}\Vert_{\infty}$ is the largest coordinate.
In the following, we show a corollary concerning the lower bound of $\delta$.
\begin{cor}
\label{best}
When $D = lcm (m_1, m_2, ... ,m_L)$,
\begin{equation}
\delta \leq \min_{S \subset  \{1,2,...,L\}} gcd( lcm(m_l, l \in S), lcm(m_l, l \in \bar{S}))
\end{equation}
\end{cor}
\begin{proof}
Let us consider a class of residue vectors, in which the residues modulo $m_l$, for $l\in S$, are set to be 0 and the others are set to be some $a<\min m_l$.
If such a residue vector exists, $X$ can be represented by $X=k_1 lcm(m_l, l \in S) = k_2 lcm(m_l, l \in \bar{S}) + a$.
Since $D =lcm( lcm(m_l, l \in S), lcm(m_l, l \in \bar{S}))$, the existence of such $X$ means that the following Diophantine equation has a solution,
\begin{equation}
k_1 lcm(m_l, l \in S) - k_2 lcm(m_l, l \in \bar{S}) =a.
\end{equation}
It is further equivalent to $gcd( lcm(m_l, l \in S), lcm(m_l, l \in \bar{S})) | a$.
Therefore, our claims hold.
\end{proof}

Clearly, when $M_l = {m_l}/{\Gamma}$ are relatively co-prime to each other, the error bound is $\delta = {\Gamma}/{4}$.
Correspondingly, $D= \Gamma \prod_{l=1}^{L} M_l$.
On the other hand, consider that the moduli are not in such form, while we still expect to achieve a same error bound $\delta= {\Gamma}/{4}$. 
From Corollary \ref{best}, for any two disjoint partitions of $\mathcal{M} = m_{[1:L]}$, the lcm of each part must share a factor no less than $\Gamma =4\delta$.
Now, we select $S$ to contain only a single modulus, which from Corollary 1 indicates that each modulus $m_l$ must share a common factor no less than $\Gamma$ with the lcm of the rest moduli.
With the consideration that $D = lcm (m_1, m_2, ... ,m_L)$, if such a common factor is not the same, more redundancy is needed compared with the case $M_l = {m_l}/{\Gamma}$ and $D= \Gamma \prod_{l=1}^{L} M_l$.

\subsection{Error Tolerance Bound and Reconstruction Scheme}
\noindent In this section, we study the relationship between robustness and the encoding dynamic range via a reconstruction of a robust decoding algorithm.
Assume that $\bm{X} = \{X_1, X_2, ... ,X_N\}$ are $N$ distinct real numbers within the dynamic range $[0,D)$ to be determined. Suppose a modulus set $\mathbb{\mathcal{M}}=\{m_1, m_2, ... ,m_L\}$ is given, where $m_l= \Gamma M_l$ with pairwise coprime $M_l$, and without loss of generality $M_{[1:L]}$ are assumed to be in an ascending order, i.e., $M_1 < M_2 <...< M_L$. We use $\mathcal{R}_l = \{ \widetilde{r}_{il} = \langle X_i +\Delta_{il} \rangle_{m_l}, i=1,2,...,N\}$ to denote the noisy observations (encoding) in each modulus $m_l\in \mathcal{M}$. Moreover,
$\widetilde{r}^c_{il}=\langle\widetilde{r}_{il}\rangle_{\Gamma}$ denotes a noisy observation of the {\em common residue} $r^c_i =\langle X_i \rangle_{\Gamma}. $

\begin{thm}
\label{thm:complex_fre}
For any $\bm{X} \in\big[ 0, \Gamma (\prod_{l=1}^{\lceil \frac{L}{N} \rceil} M_l-1)\big)^N$, if the errors underlying satisfy $\max_{il} |\Delta_{il}| < \delta = \frac{\Gamma}{4}$, there exists a robust decoding algorithm returning an estimation $\hat{\bm{X}} = \hat{X}_{[1:N]}$ such that $|\hat X_i-X_i|<\frac{3\Gamma}{4}$ for each $i$. 
\end{thm}

We present a sketch of the proof, where the details can be found in Appendix \ref{pf:thm_complex}. The decoding algorithm can be found in Algorithm \ref{alg:complex_decoding}. Roughly speaking, the main procedure can be summarized into two steps: hypothesis and test. As mentioned before, one key challenge in decoding is the unknown correspondence. Therefore, given an arbitrary clustering proposal by dividing the residues into $N$ subsets, a verification should exist to check whether any residue is {\em misclassified} \footnote{It is possible that different $X_i$ and $X_j$, or their folding numbers $\lfloor X_i/\Gamma \rfloor$ and $\lfloor X_j/\Gamma \rfloor$, share the same residue. We will take those residues interchangeable in our analysis.}. To complete the whole reconstruction, besides a correct guess of clustering, another statistic needs to be recovered is {\em the order of errors}: the ascending order of $\Delta_{il}$. It is noted that modulo operation results in {\em non-weighted} residues with ambiguity. For example, under noise perturbation, when we observe a noisy residue 11 modulo 12, it can either be $13-2$ or $9+2$. This raises another key challenge in the robust decoding that {\em the empirical average of residues is not robust.} However, the ascending order of noises $\Delta_{il}$ enables us to remove the modulo operation and simply represent the noisy residue over the real axis instead. Continued to the earlier example, for some number $X$ with a true residue $1$ modulo 12, where we have two observations $11 = \langle 1-2 \rangle_{12}$ and $2 = \langle 1+1 \rangle_{12}$. A simple average of $2$ and $11$, which gives an estimate $6.5$, produces an estimation error larger than a quarter of the modulus $12$ and we will show such trivial estimation cannot give a satisfied error control at the end. Nonetheless, given the order of errors, where we represent the noisy observation in the following equivalent form $(11-1 \times 12)$ and $(2-0\times 12)$, then a trivial CRT is indeed robust already.  
In Algorithm \ref{alg:complex_decoding}, the proposed verification is formed by two subroutines and we show that if the proposal passes through the two validity tests, then it can produce desired error control.

With the above understanding, when $\delta = \Gamma/4$, a correct proposal on both clustering and the order of noise can be essentially modeled by $S_i = \{\tilde{r}_{(il)}, l=1,2,...,L \}$ and $\tau_{(il)} \in \{0,1\}$, where $(il)$ denotes the index given the clustering, such that $S_i$ are the residues of $X_i$ and meanwhile $\{ \langle\tilde{r}_{(il)}\rangle_{\Gamma}-\tau_{(il)}\Gamma, l=1,2,...,L \}$ in an ascending order corresponds to that of $\{ \Delta_{(il)}, l=1,2,..., L\}$. The proof heavily relies on the following two facts:
\begin{enumerate}
    \item Given the correct proposal of clustering and  $\tau_{(il)}$, the shifted common residues satisfy
    $$| (\tilde{r}^c_{(il)} -\tau_{(il)})  -(\tilde{r}^c_{(il')}-\tau_{(il')}) |< {\Gamma}/{2}. $$
    \item Based on the pigeonhole principle, for any proposal of clustering, there exist at least $\lceil\frac{L}{N}\rceil$ from the $L$ residues are all from one integer $X_{i_0}$. 
\end{enumerate}

The first criterion is stemmed from the assumption that the magnitude of error is smaller than $\Gamma/4$, and thus the distance between any two noisy observations of the common residue should not be larger than $\Gamma/2$, after they have been sorted according to the order of errors.

\begin{algorithm}
\caption{Robust Remaindering Decoding of Complex Waveform}
\textbf{Input}: modulus set: $\mathcal{M}=\{m_l=M_l\Gamma|l=1,2,\dots,L\}$, where $M_{[1:L]}$ are sorted in ascending order;\\
Residue Sets: $\mathcal{R}_{l}=\{\widetilde{r}_{il}|i=1,2,\dots,N\}$, $l=1,2,\dots,L$.

\begin{algorithmic}[1]

\STATE \textbf{Repeat}: Propose a clustering assignment
\STATE Following the proposed clustering by selecting corresponding residue from each $\mathcal{R}_l$ to obtain a $L$-residue clustering $S_i=\{\widetilde{r}_{(il)}|l=1,2,\dots,L\}$, $i = 1, 2, ... ,N$. Here, $(il)$ denotes the index of the residues assigned to $S_i$.

\STATE Assign a binary parameter $\tau_{(il)} \in\{0,1\}$ to each residue in $S_i$ randomly. 
\STATE Calculate the shifted common residues $\hat{r}^c_{(il)}=\langle\widetilde{r}_{(il)}\rangle_{\Gamma}-\tau_{(il)}\Gamma$ corresponding to each $\widetilde{r}_{(il)}$ in $S_i$, for $i=1,2,...,N$.
\STATE For each $S_i$, calculate $q_i\equiv\frac{\widetilde{r}_{(il)}-\hat{r}^c_{(il)}}{\Gamma}\mod M_l$ via CRT.
\STATE\textbf{Until}: Each $q_i$ satisfies that $q_i\in[0,D_q)$ where $D_q=\prod_{l=1}^{\lceil \frac{L}{N} \rceil} M_l\Gamma$ and each residue in $S_i$ satisfies equation (\ref{key-idea}) for $i=1,2,\dots,N$.
\end{algorithmic}
\textbf{Output}: $\hat{X_i}=q_i\Gamma+{\sum^{L}_{l=1}\hat{r}^c_{(il)}}/{L}$.
\label{alg:complex_decoding}
\end{algorithm}

Before the end of this section, we give an interesting corollary that, excluding a negligible failure set, the reconstruction error bound can be further improved from $3\Gamma/4$ to $\Gamma/4$.
\begin{cor}
\label{cor:1/4-complex}
For any $X_i\in[0,\Gamma(\prod^{\lceil\frac{L}{N}\rceil}_{l=1}M_l-1))^N$, with the errors introduced in residues satisfy $\max_{il}|\Delta_{il}|<\delta=\frac{\Gamma}{4}$, if $\min_{d=0,\pm 1} |r_{i_0l}-r_{i_1l}+dM_l\Gamma|>3\Gamma$ holds for $l=1,2,\dots L$ and $1\leq i_0<i_1\leq N$, $X_i$ can be recovered robustly with error bounded by $\frac{\Gamma}{4}$. 
\end{cor}

\section{Generalization to Real Waveforms}
\label{sec:real_fre}

\subsection{Necessary and Sufficient Condition}
\label{sec:real_fre_single}
\noindent Slightly different from the complex waveform, a sinusoidal real waveform has two complex components,
\begin{equation}
x(t)=A_i\cos(2\pi f_it)=A_ie^{2\pi jf_it}+A_ie^{-2\pi jf_it},
\end{equation}
which results in two symmetric peaks over the spectrum domain, i.e., with an additional negative duplicate compared to the complex waveform case. With a same reasoning, the sufficient and necessary condition to robustly recover the frequency(ies) in a real waveform can be similarly modeled as the encoding of $N$ real numbers $\{X_1,X_2,\dots,X_N\}$ by $L$ residue sets $\mathcal{R}_l=\{\widetilde r^+_{i,l}=\langle X_i+\Delta^+_{i,l}\rangle_{m_l},\widetilde r^-_{i,l}=\langle -X_i+\Delta^-_{i,l}\rangle_{m_l}|i=1,2,\dots,N\}, $ for $l=1,2,\dots, L$, where $\Delta^{\pm}_{i,l}$ represents the error.  We still assume that $m_{[1:L]}$ are in an ascending order, where $m_l = \Gamma M_l$ with co-prime $M_{[1:L]}$. To make the problem nontrivial, it is assumed that $m_l<2\max X_i$.

We first consider the noiseless real waveform model with a single frequency, i.e., ${R}_l(X) =\{ r^{+}_{l} = \langle X \rangle_{m_l},  r^{-}_{l} = \langle -X\rangle_{m_l} \}$, is given for each sampling rate $m_l$.
For a set of moduli, $\{m_1, m_2, ..., m_L\}$, its maximum dynamic range is the maximum $D$ such that any $X \in [0,D)$ has a unique residue set representation $\mathcal{R}_l(X)$.
In this section, we provide a closed-form equation for the maximum dynamic range under the errorless real waveform model.

Let us suppose that $D$ is the maximum dynamic range given $\mathcal{M}=\{m_1, m_2, ..., m_L\}$.
By definition, there must exist another $Y \in  [0,D)$ such that $\mathcal{R}_l(D) = \mathcal{R}_l(Y)$ for  $l \in [1:L]$.
Therefore, for each $m_l$, $D$ must satisfy at least one of the following equations,
\begin{equation}
\label{4-range-case-1}
D+Y  \equiv 0 \mod m_l
\end{equation}
or
\begin{equation}
\label{4-range-case-2}
D-Y \equiv 0 \mod m_l
\end{equation}
Let $\mathbb{U}_1$ and $\mathbb{U}_2$ be two subsets of $\{m_1, m_2, ... ,m_L\}$, where $\mathbb{U}_1$ and $\mathbb{U}_2$ denote the set of the moduli satisfying (\ref{4-range-case-1}) and (\ref{4-range-case-2}), respectively.
The two sets might have an intersection and their union is $\{m_1, m_2, ... ,m_L\}$, i.e., $\mathbb{U}_1 \cup \mathbb{U}_2 = \{m_1, m_2, ... ,m_L\}$.
This yields that
\[
\left\{
            \begin{array}{lr}
            D+Y  \equiv 0 \mod lcm(m_l \in \mathbb{U}_1)\\
            D- Y  \equiv 0 \mod lcm(m_l \in \mathbb{U}_2)
             \end{array}
             \right..
\]
Since $D>Y>0$, $D+Y$ and $D-Y$ are both non-zero. Thus, $D+Y \geq lcm(m_l \in \mathbb{U}_1)$ and $D-Y \geq lcm(m_l \in \mathbb{U}_2)$, which implies that
\begin{equation}\label{equ:lowerbound}
D \geq \frac{lcm(m_l \in \mathbb{U}_1)+ lcm(m_l \in \mathbb{U}_2)}{2}.
\end{equation}
If we let $D = \frac{lcm(m_l \in \mathbb{U}_1)+lcm(m_l \in \mathbb{U}_2)}{2}$ and $Y= |\frac{lcm(m_l \in \mathbb{U}_1)-lcm(m_l \in \mathbb{U}_2)}{2}|$, $\mathcal{R}_l(D)$ is exactly the same as $\mathcal{R}_l(Y)$.
It means that equation (\ref{equ:lowerbound}) is satisfied with a tight lower bound of $D$.
Notice that $\mathbb{U}_1$ and $\mathbb{U}_2$ can be arbitrary two sets that the union of both is $\{m_1, m_2, ... ,m_L\}$.
Thus, finding the maximum dynamic range is equivalent to finding a proper subset $\mathbb{U} \subset \{m_1, m_2, ... ,m_L\}$ to minimize
\begin{equation}
\label{max-dynamic}
D =  \min_{\mathbb{U}}\frac{lcm(m_l \in \mathbb{U}) + lcm(\{m_1, m_2, ... ,m_L\} /\mathbb{U})}{2}.
\end{equation}


In the following, we give the robustness bound for the single frequency estimation. 

\begin{thm}
\label{thm:single-real}
If the error introduced in residues such that $\max|\Delta^{\pm}_{l}|<\delta=\frac{\Gamma}{4}$, $X$ can be recovered robustly error bounded by $\frac{3\Gamma}{4}$ for any $X\in[0,D)$, where \footnote{If we assume that the common residue $r_c=\langle X \rangle_{\Gamma} \not = 0 ~\text{ or }~ \frac{\Gamma}{2}$, the maximum dynamic range is equivalent to $lcm(m_1,m_2,...,m_L)$, where $\Gamma$ represents $gcd(m_1,m_2,...,m_L)$.
Such claim is based on the fact that for all $X$ satisfying our assumptions, its positive and negative residues $r^{+}_{c,l}$ and $r^{-}_{c,l}$ are not the same, i.e., $\langle r^{+}_{c,l} \rangle_{\Gamma} = \langle X \rangle_{\Gamma} 
 \not =  \langle -X \rangle_{\Gamma} = \langle r^{-}_{c,l} \rangle_{\Gamma}.$
This allows us to distinguish $r^{+}_{c,l}$ from $r^{-}_{c,l}$ and easily recover $X$.
When the errors exist in our model, if the magnitudes of errors are smaller than $ \frac{\min \{ \langle 2X \rangle_{\Gamma}, \Gamma- \langle 2X \rangle_{\Gamma}\}}{2},$ $r^{+}_{c,l}$ and $r^{-}_{c,l}$ can still be distinguished. }
$$D \leq [0,\min_{\mathbb{U} \subset \{M_1,M_2,\dots,M_L\}}(\frac{\prod_{l\in\mathbb{U}}M_l+\prod_{l\in\overline{\mathbb{U}}}M_l}{2}-1)\Gamma).$$ 
\end{thm}

The robust decoding algorithm  and proof can be found in Appendix \ref{app:single_fre_real} and \ref{app:pf_single_real}, respectively. We have to emphasize that Algorithm \ref{alg:single_real} proposed can be further improved to be of linear decoding time with the trick from Generalized CRT \cite{SPL_two_integer, TSP2018} with a slight compromise in dynamic range, where $D$ becomes $$D \leq \Gamma(\prod_{l=1}^{\lceil \frac{L}{2N} \rceil}M_l-1).$$ 

\subsection{Multiple Frequencies Estimation from Real Waveforms}
\noindent In this section, we will present a generic characterization of the robustness and encoding dynamic range in the real waveform model with $N$ frequencies. With a similar reasoning, the corresponding model can be described as follows. Given the modulus set $\mathcal{M}=\{m_l=M_l\Gamma|l=1,2,\dots,L\}$, where $M_l$ are pairwise coprime and sorted in ascending order, we aim to recover $N$ distinct real numbers $\bm{X} = \{X_1, X_2, ... ,X_N\}$ within the dynamic range $[0,D)$ with $L$ residue sets $\mathcal{R}_l = \{ \widetilde{r}^+_{i,l} = \langle X_i +\Delta^{+}_{i,l} \rangle_{m_l},\widetilde{r}^-_{i,l} = \langle -X_i +\Delta^{-}_{i,l} \rangle_{m_l}|i=1,2,...,N\}.$ Here, $\Delta^{\pm}_{i,l}$ similarly captures the noises underlying. 

Indeed, the robust decoding problem in the real waveform case can be analogously addressed with the similar idea applied in the complex waveform case. Once the noisy residues can be correctly clustered, the $N$-number robust decoding problem can be reduced to $N$ independent single-number decoding, which has been solved in Section \ref{sec:real_fre_single}. The essential challenge is that now each objective number is encoded by two symmetrically positive and negative residues with noises in each sampler, which makes decoding more complicated. 

We still apply the {\em hypothesis-then-testing} framework proposed in Section \ref{sec:comple_fre}, but in the real waveform case we view the $N$ real numbers with positive and negative copies as $2N$ numbers. Thus, the residues are clustered into $2N$ subsets, still denoted by $S_i=\{\widetilde{r}_{(i,l)}|l=1,2,\dots,L\}$, for $i=1,2,...,2N$. The two criteria can be similarly derived.

\begin{thm}
\label{thm:real_fre}
In the real waveform case, if the errors introduced in residues satisfy $\max_{i,l}|\Delta^{\pm}_{i,l}|<\delta=\frac{\Gamma}{4}$, $X_i$ can be reconstructed error bounded by $\frac{3\Gamma}{4}$ for any $X_i\in[0,D)^{N}$, where $D=\min(\frac{lcm(M_l\in \mathbb{U})+lcm(M_l\in \overline {\mathbb{U}})}{2}-1)\Gamma$ and $\mathbb{U}\cup\overline{\mathbb{U}}=\{M_1,M_2,\dots,M_{\lceil\frac{L}{N}\rceil}\}$.
\end{thm}

\section{Conclusion}
\label{conclusion}
\noindent From a (robust) encoding-then-decoding viewpoint, this paper advances the understanding of sampling theory from a new perspective: given the sparsity restriction of a set of samplers temporally or spatially, what is the necessary and sufficient condition of the signal reconstruction from acquired samples? We characterize and prove such conditions in several common parameter estimation tasks such as frequency, phase difference, DoA and Doppler shift. 
One main takeaway is that, with a random selection of sampling parameter and a negligible failure probability, it is possible that one can asymptotically sharpen the deterministic necessary sampling constraint. We leave a systematical study on such improvement to our future work. 

\bibliographystyle{ieeetr}
\bibliography{ref}

\appendices 
\section{Proof of Theorem \ref{thm:complex_fre}}
\label{pf:thm_complex}
\begin{proof}
With the fact that $\lfloor\frac{X_i}{\Gamma}\rfloor\equiv\lfloor\frac{\langle X_i\rangle_{m_l}}{\Gamma}\rfloor\equiv\lfloor\frac{\langle X_i\rangle_{m_l}-\langle X_i\rangle_{\Gamma}}{\Gamma}\rfloor\equiv\frac{r_{i,l}-r^c_{i,l}}{\Gamma}\mod M_l$, to estimate the folding number $\lfloor \frac{X_i}{\Gamma} \rfloor$, we obtain
\begin{equation}
\label{folding-number}
\langle \lfloor \frac{ \widetilde{r}_{il}}{\Gamma} \rfloor \rangle_{M_l} = \left\{
            \begin{array}{lr}
          1). \langle \lfloor \frac{X_i}{\Gamma} \rfloor \rangle_{M_l}, ~when~ r^c_i + \Delta_{il} \in[0,\Gamma) \\
          2).  \langle \lfloor \frac{X_i}{\Gamma} \rfloor -1\rangle_{M_l}, ~when~ r^c_i + \Delta_{il} \in(-\Gamma,0) \\
          3).  \langle \lfloor \frac{X_i}{\Gamma} \rfloor +1\rangle_{M_l}, ~when~ r^c_i + \Delta_{il} \in[\Gamma,2\Gamma)
            \end{array}
             \right.
\end{equation}
Since $2\delta=\frac{\Gamma}{2}$, the three mentioned cases cannot happen simultaneously. But case 1) and 2), or case 1) and 3), can happen at the same time, which means there could exist $l_1, l_2 \in \{1,2,...,L\}$ such that $r^c_i + \Delta_{il_1} \in [0,\Gamma)$ and $r^c_i + \Delta_{il_2} \in (-\Gamma,0)$. Therefore, $\langle \frac{ \widetilde{r}^{c}_{il_1}}{\Gamma} \rangle_{M_{l_1}}$ and $\langle \frac{ \widetilde{r}^{c}_{il_2}}{\Gamma} \rangle_{M_{l_2}}$ can be residues of either one of $\{ \lfloor \frac{X_i}{\Gamma} \rfloor, \lfloor \frac{X_i}{\Gamma} \rfloor+1, \lfloor \frac{X_i}{\Gamma} \rfloor-1 \}$. Simply aggregating them via CRT will bring unpredictable reconstruction errors. Without loss of generality, we only consider that case 1) and 2) occur in the following.

To this end, following \cite{multinumrcrt}, we consider figuring out the order of $l$ such that $\Delta_{il}$ are sorted in an ascending order. Said another way, it is equivalent to determining a binary parameter $\tau_{il} \in \{0,1\}$ such that the order of $l$, where $\hat{r}^c_{il}=\widetilde{r}^{c}_{il}-\tau_{il}\Gamma$ are in ascending order, is the same as that where $\Delta_{il}$ are sorted non-decreasingly. If so, $\langle \lfloor \frac{\widetilde{r}_{il}+\tau_{il} \Gamma}{\Gamma}  \rfloor  \rangle_{M_l} = \langle  \frac{\widetilde{r}_{il}-\widetilde{r}^c_{il}}{\Gamma} + \tau_{il} \rangle_{M_l}$ are the residues of one integer, one of $\{\lfloor \frac{X_i}{\Gamma} \rfloor,\lfloor \frac{X_i}{\Gamma}\rfloor+1,\lfloor \frac{X_i}{\Gamma}\rfloor-1\}$.
Clearly, with proper $\tau_{il}$, the problem can be addressed by the generalized CRT \cite{sharpened}. 
Furthermore, $|(\widetilde{r}^c_{il_1}-\tau_{il_1}\Gamma)-(\widetilde{r}^c_{il_2}-\tau_{il_2}\Gamma)|$ becomes the minimum distance between $\widetilde{r}^c_{il_1}$ and $\widetilde{r}^c_{il_2}$ in the ring of length $\Gamma$.
Thus, for any $l_1, l_2$,
  \begin{equation}
  \label{key-idea}
  | \widetilde{r}^c_{il_1}-\tau_{il_1}\Gamma - (\widetilde{r}^c_{il_2}-\tau_{il_2}\Gamma ) | < 2\delta = \frac{\Gamma}{2}
  \end{equation}
Equation (\ref{key-idea}), called the {\em first criterion} in the following, is one of the keys in the rest of the proof, which is also overlooked in previous works.

Now, we prove the following fact: {\em Given a selection of $\{\tau_{il}\}$, if they satisfy the first criterion, the shifted residues $\{ \langle \lfloor \frac{\widetilde{r}_{il}+\tau_{il}\Gamma}{\Gamma}\rfloor \rangle_{M_l}, l=1,2,...,L\}$ are all residues of one of $\{\lfloor \frac{X_i}{\Gamma} \rfloor, \lfloor \frac{X_i}{\Gamma} \rfloor+1, \lfloor \frac{X_i}{\Gamma} \rfloor-1 \}$.}

For an arbitrary selection of $\tau_{il}$ assigned, let $\mathbb{U}_1 \in \{1,2,...,L\}$, where case 1) occurs for $l_1 \in \mathbb{U}_1$. Correspondingly, $\mathbb{U}_2 = \overline{\mathbb{U}}_1$, where case 2) happens for $l_2 \in \mathbb{U}_2$. 

Analogously, we also specify the indices of $M_l$, where $\tau_{il}$ are set to $0$ or $1$, respectively.  We use $\mathbb{U}_{11} \subset \mathbb{U}_1$ to denote those $l_{11} \in \mathbb{U}_{11}$ such that $\tau_{il_{11}}=0$. Similarly, $l_{12} \in \mathbb{U}_{12} \subset \mathbb{U}_1,$ where $\tau_{il_{12}}=1$; $l_{21} \in \mathbb{U}_{21} \subset \mathbb{U}_2$, where $\tau_{il_{22}}=0$, and $l_{22} \in \mathbb{U}_{22} \subset \mathbb{U}_2$, $\tau_{il_{22}}=1$. Then, we have the following observations for the four subcases, respectively: (1). $\langle  \frac{\widetilde{r}_{il_{11}}-\widetilde{r}^c_{il_{11}} }  {\Gamma} + \tau_{il_{11}} \rangle_{M_l} = \langle \frac{X_i}{\Gamma} \rangle_{M_{l_{11}}}$; (2). $\langle \frac{\widetilde{r}_{il_{12}}-\widetilde{r}^c_{il_{12}} } {\Gamma}  + \tau_{il_{12}}  \rangle_{M_l} = \langle \frac{X_i}{\Gamma}+1 \rangle_{M_{l_{12}}}$;
(3). $\langle\frac{\widetilde{r}_{il_{21}}-\widetilde{r}^c_{il_{21}}  } {\Gamma}+ \tau_{il_{21}}\rangle_{M_l} = \langle \frac{X_i}{\Gamma} -1 \rangle_{M_{l_{21}}}$; (4). $ \langle\frac{\widetilde{r}_{il_{22}}-\widetilde{r}^c_{il_{22}}  } {\Gamma}+ \tau_{il_{22}} \rangle_{M_l} = \langle \frac{X_i}{\Gamma} \rangle_{M_{l_{22}}}$.

For those $\langle \frac{\widetilde{r}_{il}-\widetilde{r}^c_{il} } {\Gamma}+\tau_{il} \rangle_{M_l}$, we assume that there exists a solution of $N$ integers $\{q_1, q_2, ... ,q_N\}$ within the dynamic range, i.e., $q_i\equiv \frac{\widetilde{r}_{k_l(i)l}-\widetilde{r}^c_{k_l(i)l} } {\Gamma}+\tau_{k_l(i)l}\mod M_l$. Here, $k_l(i)$ is a permutation on $\{1,2,...,N\}$, indicating the correspondence between $\widetilde{r}_{k_l(i)l}$ and $q_i$ for each $M_l$. 
We elaborate on the wrong residue classifications below. 
\begin{itemize}
\item There exists $i \in \{1,2,...,N\}$ such that $k_{l_{11}}(i) = i_0$ and $k_{l_{12}}(i) = i_0$.
Then, $| \widetilde{r}^c_{i_0l_{11}}-\tau_{i_0l_{11}}\Gamma -  (\widetilde{r}^c_{i_0l_{12}}- \tau_{i_0l_{12}}\Gamma) |=|\widetilde{r}^c_{i_0l_{11}}- \widetilde{r}^c_{i_0l_{12}}+\Gamma| = |r^c_{i_0}+\Delta_{i_0l_{11}}-r^c_{i_0}-\Delta_{i_0l_{12}}+\Gamma| > \Gamma - 2\delta = \frac{\Gamma}{2}$

\item There exists $i \in \{1,2,...,N\}$ such that $k_{l_{11}}(i) = i_0$ and $k_{l_{21}}(i) = i_0$.
Then $|\widetilde{r}^c_{i_0l_{11}}- \tau_{il_{11}}\Gamma -(\widetilde{r}^c_{i_0l_{21}}- \tau_{il_{21}}\Gamma) |=|\widetilde{r}^c_{i_0l_{11}}- \widetilde{r}^c_{i_0l_{21}}| = |r^c_{i_0}+\Delta_{i_0l_{11}} - (r^c_{i_0}+\Delta_{i_0l_{21}}+\Gamma) | > \frac{\Gamma}{2}$

\item There exists $i \in \{1,2,...,N\}$ such that $k_{l_{21}}(i) = i_0$ and $k_{l_{22}}(i) = i_0$.
Then $|\widetilde{r}^c_{i_0l_{21}}- \tau_{il_{21}}\Gamma -(\widetilde{r}^c_{i_0l_{22}}- \tau_{il_{22}}\Gamma) |=  |r^c_{i_0}+\Delta_{i_0l_{21}} +\Gamma - (r^c_{i_0}+\Delta_{i_0l_{22}}+\Gamma-\Gamma) |> \frac{\Gamma}{2}$

\item There exists $i \in \{1,2,...,N\}$ such that $k_{l_{12}}(i) = i_0$ and $k_{l_{22}}(i) = i_0$.
Then $|\widetilde{r}^c_{i_0l_{12}}- \tau_{il_{12}}\Gamma -(\widetilde{r}^c_{i_0l_{22}}- \tau_{il_{22}}\Gamma) |=  |r^c_{i_0}+\Delta_{i_0l_{12}} -
\Gamma - (r^c_{i_0}+\Delta_{i_0l_{22}}+\Gamma-\Gamma) |> \frac{\Gamma}{2}$

\item There exists $i \in \{1,2,...,N\}$ such that $k_{l_{12}}(i) = i_0$ and $k_{l_{21}}(i) = i_0$.
Then $|\widetilde{r}^c_{i_0l_{12}}- \tau_{il_{12}}\Gamma -(\widetilde{r}^c_{i_0l_{21}}- \tau_{il_{21}}\Gamma) |=  |r^c_{i_0}+\Delta_{i_0l_{12}} -
\Gamma - (r^c_{i_0}+\Delta_{i_0l_{21}}+\Gamma) |> \frac{\Gamma}{2}$
\end{itemize}
Therefore, any incorrect selection of $\tau_{il}$ in $S_i$ will result in a contradiction to the {\em first criterion}.

Next, we state the {\em second criterion} that the recovered $\{q_1,q_2,\dots,q_N\}$ should be within the dynamic range $[0,D_q)$, where $D_q=\prod_{l=1}^{\lceil\frac{L}{N}\rceil} M_l.$ In the following, we prove that a proper selection of $\tau_{il}$ and clustering satisfying both criteria are sufficient for a robust reconstruction.   

For $X_i<(\prod^{\lceil\frac{L}{N}\rceil}_{l=1}M_l-1)\Gamma$, we obtain $\lfloor\frac{X_i}{\Gamma}\rfloor+1<\prod^{\lceil\frac{L}{N}\rceil}_{l=1}M_l$. If the residues of $q_i$ are from one integer, we have  $q_i\in\{\lfloor\frac{X_i}{\Gamma}\rfloor,\lfloor\frac{X_i}{\Gamma}\rfloor+1,\lfloor\frac{X_i}{\Gamma}\rfloor-1\}$, i.e., $q_i<\prod^{\lceil\frac{L}{N}\rceil}_{l=1}M_l\Gamma$.
We assume that there exists another solution $\{q_1',q_2',\dots,q_N'\}$. Then, for the residue classification of $q'_1$, there must be at least $\lceil \frac{L}{N} \rceil$ residues of $q'_1$ are from one of $\{\lfloor\frac{X_i}{\Gamma}\rfloor,\lfloor\frac{X_i}{\Gamma}\rfloor+1,\lfloor\frac{X_i}{\Gamma}\rfloor-1\}$. On the other hand, $q'_1$ cannot share the residues of $\lfloor\frac{X_{i_0}}{\Gamma}\rfloor$, $\lfloor\frac{X_{i_0}}{\Gamma}\rfloor+1$ and $\lfloor\frac{X_{i_0}}{\Gamma}\rfloor-1$ simultaneously. Otherwise, the first criterion is violated. Thus, without loss of generality, we assume there exist $\lceil\frac{L}{N}\rceil$ residues of $q'_1$ are from $\lfloor\frac{X_{i_0}}{\Gamma}\rfloor$, i.e., $q'_1\equiv \lfloor\frac{X_{i_0}}{\Gamma}\rfloor\mod lcm(M_l\in \mathbb{U}')$, where $|\mathbb{U}'|=\lceil\frac{L}{N} \rceil$.
Thus, $q'_1=klcm(M_l\in \mathbb{U}')+\lfloor\frac{X_{i_0}}{\Gamma}\rfloor$, where $k\geq 0$.
With our assumption of $D_q$, if $k=0$, $q'_1 = \lfloor\frac{X_{i_0}}{\Gamma}\rfloor$. Otherwise, $k\geq 1$, $q'_1>D_q$, a contradiction. Now, after determining the folding numbers $q_i$, we obtain
\begin{equation}
X_i=q_i\Gamma +\frac{\sum^{l=L}_{l=1}\hat r_{k_l(i)l}}{L}.
\end{equation}

Finally, we determine the worst-case reconstruction error. From the above discussions, with a correct selection of $\tau_{il}$ and clustering such that $S_i$ are all residues of $X_i$, which clearly satisfy the two criteria, and then the reconstruction error is upper bounded by the maximal magnitude of $\Delta_{il}$. However, there is a special case that the folding number of different $X_i$ and $X_j$ may share the same residues: there exist some $l \in \{1,2,...,L\}$ such that $i \not = j \in \{1,2,...,N\}$, $q_{i} \equiv q_{j} \mod M_l$. They are essentially interchangeable, which will not lead to the failure to recover $q_i$, but we have to quantify the reconstruction error carefully.Since $\hat{r}^c_{jl}$ is assigned as a common residue to $X_{i}$, i.e., the common residues assigned to $X_{i}$ are a combination of $\hat{r}^c_{il_1}$ and $\hat{r}^c_{jl_2}$.
Without loss of generality, we assume that $q_{i_0}=\lfloor\frac{X_{i_0}}{\Gamma}\rfloor-1$, which means there are at least $\lceil\frac{L}{N}\rceil$ common residues are from $X_{i_0}$.
It must fall into case 2) and $\tau_{i_0l}=0$.
Since $r^c_{i_0}+\Delta_{i_0l}<0$, $\hat{r}^c_{i_0}=\langle\widetilde{r}^c_{i_0l}\rangle_{\Gamma}=r^c_{i_0}+\Delta_{i_0l}+\Gamma$.
The rest common residues satisfy $\max|\hat{r}^c_{il}-\hat{r}^c_{i_0l}|<\frac{\Gamma}{2}$.
Then, $|\hat{X}_{i_0}-X_{i_0}|$ is equal to
\begin{equation}
\begin{aligned}
&|(\lfloor\frac{X_{i_0}}{\Gamma}\rfloor-1)\Gamma+\frac{\sum^{l=\lceil\frac{L}{N}\rceil}_{l=1}\hat{r}^c_{i_0l}+\sum^{l=L}_{l=\lceil\frac{L}{N}\rceil+1}\hat{r}^c_{il}}{L}-X_{i_0}|\\
&<|\frac{\sum^{l=\lceil\frac{L}{N}\rceil}_{l=1}\max\Delta_{i_0l}+\sum^{l=L}_{l=\lceil\frac{L}{N}\rceil+1}(\max\Delta_{i_0l}+\frac{\Gamma}{2})}{L}|\\
&<|\frac{\lceil\frac{L}{N}\rceil\frac{\Gamma}{4}+(L-\lceil\frac{L}{N}\rceil)\frac{3\Gamma}{4}}{L}|<\frac{3\Gamma}{4}.
\end{aligned}
\end{equation}
Therefore, the reconstruction is error-bounded by $\frac{3\Gamma}{4}$.
\end{proof}

\section{Proof of Corollary \ref{cor:1/4-complex}}
\begin{proof}
Likewise, let $q_i$ denote the folding number estimations, where $q_i\in\{\lfloor\frac{X_i}{\Gamma}\rfloor,\lfloor\frac{X_i}{\Gamma}\rfloor+1,\lfloor\frac{X_i}{\Gamma}\rfloor-1\}$. It is noted that $\min_{d \in  \{0, \pm 1\}}|r_{i_0l}-r_{i_1l}+dM_l\Gamma|$ represents the minimum distance between $r_{i_0l}$ and $r_{i_1l}$ on the circle of length $M_l\Gamma$. Then, we have the following facts and at least one of the three cases happens: 1). $r_{il}=\langle q_i\Gamma\rangle_{M_l\Gamma}+r^c_i$, when $q_i=\lfloor\frac{X_i}{\Gamma}\rfloor$; 2). $r_{il}=\langle q_i\Gamma\rangle_{M_l\Gamma}+r^c_i-\Gamma$, when $q_i= \lfloor\frac{X_i}{\Gamma}\rfloor+1$; 3). $r_{il}=\langle q_i\Gamma\rangle_{M_l\Gamma}+r^c_i+\Gamma$, when $q_i= \lfloor\frac{X_i}{\Gamma}\rfloor-1$.

When $r_{i_0l}$ and $r_{i_1l}$ satisfy 2) and 3) respectively,
since $\min|r_{i_0l}-r_{i_1l}+dM_l\Gamma|>3\Gamma$, replacing $r_{i_0l}$ and $r_{i_1l}$ with the right hand of 2) and 3), we have
\begin{equation}
    \min|\langle q_{i_0}\Gamma\rangle_{M_l\Gamma}-\langle q_{i_1}\Gamma\rangle_{M_l\Gamma}+dM_l\Gamma+(r^c_{i_0}-\Gamma-r^c_{i_1}-\Gamma)|>3\Gamma,
\end{equation}
and thus $|r^c_{i_0}-r^c_{i_1}-2\Gamma|<3\Gamma$, which results in $q_{i_0}\not\equiv q_{i_1}\mod M_l$.
The same conclusion can be derived when $r_{i_0l}$ and $r_{i_1l}$ satisfy any two equations of 1), 2) and 3), i.e., for each $l$, $q_{i_0}\not\equiv q_{i_1} \mod M_l$ holds.
If $\hat{r}^c_{i_0l}$ and $\hat{r}^c_{i_1l}$ are clustered into one set $S_i$, the estimated folding number $q_i$ satisfies
\begin{equation}
   \left\{
            \begin{array}{lr}
           q_i\equiv q_{i_0}\mod lcm(M_l\in\mathbb{U}_{i_0})\\
           q_i\equiv q_{i_1}\mod lcm(M_l\in\mathbb{U}_{i_1}).
            \end{array}
             \right. 
\end{equation}
Based on the pigeonhole principle, at least $\lceil\frac{L}{N}\rceil$ residues in $S_i$ are from one integer.
Without loss of generality, we assume that $|\mathbb{U}_{i_0}|\geq\lceil\frac{L}{N}\rceil$.
Therefore, $q_i=klcm(M_l\in\mathbb{U}_{i_0})+q_{i_0}$. Because $q_{i_0}\not\equiv q_{i_1}\mod M_l$, we have $q_i\not=q_{i_0}$ and $k>0$, which leads to a contradiction to the {\em second criterion}.
That is to say, if for each $l$, $q_{i_0}\not\equiv q_{i_1} \mod M_l$, clustering residues from different integers violates the {\em second criterion}.
So, for each $q_i$, the residues must be all from $X_i$, which provides a sharpened robust reconstruction error bounded by $\frac{\Gamma}{4}$. \footnote{If $r_{il}$ is assumed to be uniformly distributed across $[0,M_l\Gamma)$, we have $\Pr(\min_{d}|r_{i_0l}-r_{i_1l}+dM_l\Gamma|>3\Gamma)= \prod_{l=1}^{L}\frac{M_l-6}{M_l}$.}
\end{proof}

\section{Closed-form Robust Reconstruction for Single Tone Real Waveforms}
\label{app:single_fre_real}
\noindent To be self-contained, we first introduce the setup and notations. Assume that the moduli are in a form $\{m_l = \Gamma M_l|l=1,2,\dots,L\}$, where $\{M_l\}$ are relatively co-prime, ranging in an ascending order. The dynamic range of $X$ is represented by $D$, i.e., $X\in [0,D)$, where $D\leq\min_{\mathbb{U} \subset \{M_1,M_2,\dots,M_L\}}(\frac{\prod_{l\in\mathbb{U}}M_l+\prod_{l\in\overline{\mathbb{U}}}M_l}{2}-1)\Gamma$. The residue set is $\mathcal{R}_l=\{\widetilde{r}^{+}_l=\langle X+\Delta^{+}_l\rangle_{m_l},\widetilde{r}^{-}_l=\langle -X+\Delta^{-}_l\rangle_{m_l}\}$, where $|\Delta^{\pm}_l| <\delta = \frac{\Gamma}{4}$.

With a similar idea, we consider selecting a residue from each $\mathcal{R}_l$ randomly to obtain a $L$-residue clustering $S$, which possibly contains residues from $X$ and $-X$ simultaneously. To recover the folding number $\lfloor\frac{X}{\Gamma}\rfloor$ robustly, we apply (\ref{op1}) and (\ref{op2}) on residues respectively. At least one of the operations leads to residues sorted in ascending order, i.e., residues are sorted in the same order as that of $\Delta^{\pm}_{l}$. Clearly, there exists an efficient solution when all residues in $S$ are from $X$ and are properly sorted. Otherwise, as proved in the next section, one can distinguish that, where at least one of the two criteria, (\ref{test1-real}) and (\ref{folding-number-range}), will not hold.

\begin{algorithm}
\caption{Robust Remaindering Decoding of Single Tone Real Waveform}
\textbf{Input}: Moduli set: $\mathcal{M}=\{m_l=M_l\Gamma|l=1,2,\dots,L\}$;\\
Residue Sets: $\mathcal{R}_{l}=\{\widetilde{r}^{+}_{l}, \widetilde{r}^{-}_{l}\}$, $l=1,2,\dots,L$.

\begin{algorithmic}[1]

\STATE \textbf{Repeat}: Propose a clustering assignment
\STATE Following the proposed clustering by selecting one residue from each $\mathcal{R}_l$ to obtain a $L$-residue clustering $S=\{\widetilde{r}_{(l)}|l=1,2,\dots,L\}$, where $(l)$ denotes the index of the residues assigned to $S$.
\STATE For each $\widetilde{r}_{(l)}$ in $S$, calculate the corresponding shifted common residues $\hat{r}_{c,(l)}$ based on (\ref{op1}) and (\ref{op2}) respectively. We obtain the shifted common residue set $S_{c1}$ and $S_{c2}$.

\STATE For $S_{c1}$ and $S_{c2}$, calculate $q\equiv\frac{\widetilde{r}_{(l)}-\hat{r}_{c,(l)}}{\Gamma}\mod M_l$ via CRT respectively.
\STATE\textbf{Until}: For $S_{c1}$ or $S_{c2}$, if each shifted common residue satisfies (\ref{test1-real}) and the corresponding $q$ satisfies that $q\in[0,D_q)$, where $D_q=\min_{\mathbb{U}\subset\{M_1,M_2,\dots,M_L\}}\frac{\prod_{l\in\mathbb{U}}M_l+\prod_{l\in\overline{\mathbb{U}}}M_l}{2}-1$, we output the one that passes the two tests. 
\end{algorithmic}
\textbf{Output}: $\hat{X}=q\Gamma+\frac{\sum^{L}_{l=1}\hat{r}_{c,(l)}}{L}$.
\label{alg:single_real}
\end{algorithm}

\section{Proof of Theorem \ref{thm:single-real}}
\label{app:pf_single_real}

\begin{proof}
With the fact that, $r^{-}_c=\Gamma-r^{+}_c$, it is clear that $\{r^{+}_c, r^{-}_c\}$ must fall into one of the following four cases, 
\begin{enumerate}
\centering
 \item$r^{+}_c \in [{\Gamma}/{4},{\Gamma}/{2})$ and $r^{-}_c \in ({\Gamma}/{2},{3\Gamma}/{4}]$
 \item$r^{+}_c \in [{\Gamma}/{2},{3\Gamma}/{4})$ and $r^{-}_c \in ({\Gamma}/{4},{\Gamma}/{2}]$
 \item $r^{+}_c \in [0,{\Gamma}/{4})$ and $r^{-}_c \in ({3\Gamma}/{4},\Gamma]$
 \item $r^{+}_c \in [{3\Gamma}/{4},\Gamma)$ and $r^{-}_c \in (0,{\Gamma}/{4}]$ 
\end{enumerate}
Recall (\ref{folding-number}), to ensure robustness, we consider applying $\tau^{+}_{l}\in\{0,1\}$ to determine the order of $l$ such that $\Delta^{+}_{l}$ are in ascending order, i.e., $\hat{r}^{+}_{c,l}=\widetilde{r}^{+}_{c,l}-\tau^{+}_{l}\Gamma$ are sorted in the order of $\Delta^{+}_{l}$. Here, we use $\widetilde{r}^{+}_{c,l} = \langle \widetilde{r}^{+}_{l} \rangle_{\Gamma}$, the noisy observation of common residue from the $l$-th sampler. Similarly, $\hat{r}^{-}_{c,l}=\widetilde{r}^{-}_{c,l}-\tau^{-}_{l}\Gamma$ can be sorted based on the same idea. To this end, we consider the following operations on the residues if we can distinguish cases (1,2) and (3,4). 
\begin{itemize}
\item operation 1: 
\begin{equation}
\label{op1}
\left\{
            \begin{array}{lr}
             \hat{r}^{+}_{c,l}=\widetilde{r}^{+}_{c,l} \\
             \hat{r}^{-}_{c,l}=\widetilde{r}^{-}_{c,l}
            \end{array}
             \right.
\end{equation}
\item operation 2: 
\begin{equation}
\label{op2}
\left\{
            \begin{array}{lr}
             \hat{r}^{+}_{c,l}=\widetilde{r}^{+}_{c,l} ~ \text{when} ~ \widetilde{r}^{+}_{c,l} \in [0,\frac{\Gamma}{2}),~ \text{or}~ \hat{r}^{+}_{c,l}=\widetilde{r}^{+}_{c,l}-\Gamma\\
           \hat{r}^{-}_{c,l}=\widetilde{r}^{-}_{c,l} ~ \text{when} ~ \widetilde{r}^{-}_{c,l} \in [0,\frac{\Gamma}{2}),~ \text{or}~ \hat{r}^{-}_{c,l}=\widetilde{r}^{-}_{c,l}-\Gamma
            \end{array}
             \right.
\end{equation}
\end{itemize}

To arrange $\widetilde{r}^{\pm}_{c,l}$ in an ascending order in cases (1,2), since $\delta = \frac{\Gamma}{4}$, we will apply operation 1, defined in (\ref{op1}), on $\widetilde{r}^{\pm}_{c,l}$, where $\hat{r}^{\pm}_{c,l} = \widetilde{r}^{\pm}_{c,l} = {r}^{\pm}_c + \Delta^{\pm}_l$. Consequently, $(\langle \frac{ \widetilde{r}^{+}_{l} -  \hat{r}^{+}_{c,l}}{\Gamma} \rangle_{M_l},\langle \frac{ \widetilde{r}^{-}_{l} -  \hat{r}^{-}_{c,l}}{\Gamma} \rangle_{M_l})$ are residues of $(Y,-Y -1)$ modulo $M_l$, where $Y$ denotes the folding number $\lfloor\frac{X}{\Gamma}\rfloor$.

Similarly, to sort $\widetilde{r}^{\pm}_{c,l}$ non-decreasingly in case 3), since $r^{+}_{c,l} + \Delta^{+}_l \in (-\frac{\Gamma}{4},\frac{\Gamma}{2})$ and $r^{-}_{c,l} + \Delta^{-}_l \in (\frac{\Gamma}{2}, \frac{5\Gamma}{4})$, by applying operation 2, shown in (\ref{op2}), $(\langle \frac{ \widetilde{r}^{+}_{l} -  \hat{r}^{+}_{c,l}}{\Gamma} \rangle_{M_l},\langle \frac{ \widetilde{r}^{-}_{l} -  \hat{r}^{-}_{c,l}}{\Gamma} \rangle_{M_l})$ are residues of $(Y,-Y)$ modulo $M_l$.

In case 4), which is the dual circumstance of case 3), $(\langle \frac{ \widetilde{r}^{+}_{l}-  \hat{r}^{+}_{c,l}}{\Gamma} \rangle_{M_l}, \langle \frac{ \widetilde{r}^{-}_{l} -  \hat{r}^{-}_{c,l}}{\Gamma} \rangle_{M_l})$ are residues of $(Y +1,-Y-1 )$ modulo $M_l$.

Provided sorted $\hat{r}^{\pm}_{c,l}$, since $\delta=\frac{\Gamma}{4}$, we have 
\begin{equation}
\label{test1-real}
    |\hat{r}_{c,l_1}-\hat{r}_{c,l_2}|< 2\delta=\frac{\Gamma}{2},
\end{equation}
where $\hat{r}_{c,l_1}$ and $\hat{r}_{c,l_2}$ denote the shifted common residues derive from one integer. Similar to the complex waveform case, (\ref{test1-real}) behaves as a criterion for a valid test.
However, the main obstacle is the misuse of operation 1 or 2, i.e., applying operation 1 on cases (3,4) or operation 2 on cases (1,2). For example, if operation 1 is applied in case 3, $ \langle \frac{ \widetilde{r}^{+}_l -  \hat{r}^{+}_{c,l}}{\Gamma} \rangle_{M_l}$ can be either $\langle Y \rangle_{M_l}$ or $\langle Y-1 \rangle_{M_l}$; and $ \langle \frac{ \widetilde{r}^{-}_l -  \hat{r}^{-}_{c,l}}{\Gamma} \rangle_{M_l}$ can be either $\langle -Y-1 \rangle_{M_l}$ or $\langle -Y \rangle_{M_l}$, which could produce an unpredictable reconstruction error. Indeed, it is impossible to distinguish the four cases just from the locations of $\widetilde{r}^{\pm}_{c,l}$.
But anyway, applying the two operations on $\widetilde{r}^{\pm}_{c,l}$, at least one of them can correctly recover $Y$. The rest proof is developed by two parts. First, if one happens to apply the proper operation, a robust estimation can be achieved. Second, if a wrong operation is applied, one can distinguish that.

Now, we prove that with correct operation, $X\in[0,D)$ is a sufficient condition that the folding number has a unique representation by residues attained, i.e., $\{Y,-Y-1\}$, $\{Y,-Y\}$ and $\{Y+1,-Y-1\}$ can be uniquely determined. Based on the conclusion from Section \ref{sec:real_fre_single}, for $\{Y, -Y\}$, the largest dynamic range of $Y$ is $\min_{\mathbb{U} \subset \{M_1, M_2, ... ,M_L\}}\frac{\prod_{l \in \mathbb{U}} M_l + \prod_{l \in \overline{\mathbb{U}}} M_l}{2}$. Analogously, the residue representation of $\{Y+1,-Y-1\}$ is unique once $Y < \min_{\mathbb{U} \subset \{M_1, M_2, ... ,M_L\}} \frac{\prod_{l \in \mathbb{U}} M_l + \prod_{l \in \overline{\mathbb{U}}} M_l}{2} -1$.
Similarly, for $\{Y, -Y-1\}$, the largest dynamic range is $\min_{\mathbb{U}\subset \{M_1, M_2, ... ,M_L\}} \ \frac{\prod_{l \in \mathbb{U}} M_l + \prod_{l \in \overline{\mathbb{U}}} M_l-1}{2}$.
As a summary, the folding number $Y$ has a unique residue
representation when 
\begin{equation}
\label{folding-number-range}
    Y < D_q=\min_{\mathbb{U} \subset \{M_1, M_2, ... ,M_L\}} \frac{\prod_{l \in \mathbb{U}} M_l + \prod_{l \in \overline{\mathbb{U}}} M_l}{2} -1.
\end{equation}
Since $X\in[0,\min_{\mathbb{U}\subset\{M_1,M_2,\dots,M_L\}}(\frac{\prod_{l\in\mathbb{U}}M_l+\prod_{l\in\overline{\mathbb{U}}}M_l}{2}-1)\Gamma)$, it's clear that $Y<D_q$, satisfying (\ref{folding-number-range}).

In the following, we switch to the second part and take a closer look at the scenario when a wrong operation is applied. To avoid tedious discussion, we only elaborate on case 2) in the following proof.
We list all the possible situations below. Let $\mathbb{U}_{11}$ denote those $l_{11}\in\mathbb{U}_{11}$ such that $r^{+}_{c}+\Delta^{+}_{l_{11}}\geq\frac{\Gamma}{2}$. Similarly, when $l_{12}\in\mathbb{U}_{12}$, $r^{+}_{c}+\Delta^{+}_{l_{12}}<\frac{\Gamma}{2}$; for $l_{21}\in\mathbb{U}_{21}$, $r^{-}_{c}+\Delta^{-}_{l_{21}}\geq\frac{\Gamma}{2}$; for $l_{22}\in\mathbb{U}_{22}$, $r^{-}_{c}+\Delta^{-}_{l_{22}}<\frac{\Gamma}{2}$. The wrong residue classifications violating (\ref{test1-real}) are elaborated below.
\begin{itemize}
    \item If $\hat{r}^{+}_{c,l_{11}}$ and $\hat{r}^{+}_{c,l_{12}}$ are clustered together, $|\hat{r}^{+}_{c,l_{11}}-\hat{r}^{+}_{c,l_{12}}|=|r^{+}_c+\Delta_{l_{11}}-\Gamma-(r^{+}_c+\Delta_{l_{12}})|=|\Gamma-(\Delta_{l_{11}}-\Delta_{l_{12}})|>\frac{\Gamma}{2}$.
    \item If $\hat{r}^{-}_{c,l_{21}}$ and $\hat{r}^{-}_{c,l_{22}}$ are clustered together, $|\hat{r}^{-}_{c,l_{21}}-\hat{r}^{+}_{c,l_{22}}|=|r^{-}_c+\Delta_{l_{21}}-\Gamma-(r^{-}_c+\Delta_{l_{22}})|=|\Gamma-(\Delta_{l_{21}}-\Delta_{l_{22}})|>\frac{\Gamma}{2}$.
    \item If $\hat{r}^{+}_{c,l_{12}}$ and $\hat{r}^{-}_{c,l_{21}}$ are clustered together, since $\hat{r}^{+}_{c,l_{12}}\in(\frac{\Gamma}{4},\frac{\Gamma}{2})$ and $\hat{r}^{-}_{c,l_{21}}\in(-\frac{\Gamma}{2},-\frac{\Gamma}{4})$,  $|\hat{r}^{+}_{c,l_{12}}-\hat{r}^{-}_{c,l_{21}}|>\frac{\Gamma}{2}$.
\end{itemize}
Therefore, under the restriction of (\ref{test1-real}), the estimated folding number $q$ must be in one of the following forms in case 2) applied with operation 2.
\begin{itemize}
\item [1] $q \equiv Y+1 \mod lcm( M_l \in \mathbb{U}_{11})$ and $q \equiv -Y-1 \mod lcm( M_l \in \mathbb{U}_{22})$ 
\item [2] $q \equiv Y+1 \mod lcm( M_l \in \mathbb{U}_{11})$ and $q \equiv -Y \mod lcm( M_l \in \mathbb{U}_{21})$
\item [3] $q \equiv Y \mod lcm( M_l \in \mathbb{U}_{12})$ and $q \equiv -Y-1 \mod lcm( M_l \in \mathbb{U}_{22})$
\end{itemize}

With the above understanding, finally, we show $\hat{X}$ can be recovered error-bounded by $\frac{3\Gamma}{4}$.
Obviously, $q$ must be in the form under operation 1:
\begin{itemize}
    \item $q \equiv Y \mod lcm( M_l \in \mathbb{U}_{1})$ and
$q \equiv -Y-1 \mod lcm( M_l \in \mathbb{U}_{2})$.
\end{itemize}
The lower bound of $\max\{lcm(M_l\in\mathbb{U}_{1}),lcm(M_l\in\mathbb{U}_{2})\}$ is $\omega = \min_{\mathbb{U}\subset\{M_1,M_2,\dots,M_L\}} \max \{\prod_{l\in\mathbb{U}}M_l,\prod_{l\in\overline{\mathbb{U}}}M_l \}$, which is bigger than $D_q$. Without loss of generality, we assume that $lcm( M_l \in \mathbb{U}_{1}) \geq lcm( M_l \in \mathbb{U}_{2})$. Thus,
\begin{equation}
\label{op1-form}
q = k \cdot lcm( M_l \in \mathbb{U}_{1}) + Y, k \geq 0
\end{equation}
If $Y \not \equiv -Y-1 \mod lcm( M_l \in \mathbb{U}_{2})$ and $\mathbb{U}_2\not=\emptyset$, $k\geq 1$ and $q$ clearly exceeds $D_q$, a contradiction. 
Only when $Y \equiv -Y-1 \mod lcm( M_l \in \mathbb{U}_{2})$ or $\mathbb{U}_2=\emptyset$, can $q$ be within the dynamic range, i.e., $k=0$ and $q=Y$.
Thus, the shifted common residues assigned to $q$ are $\hat{r}^{+}_{c,l_1}$ and $\hat{r}^{-}_{c,l_2}$, where $l_1\in\mathbb{U}_1$ and $l_2\in\mathbb{U}_2$. It is worth pointing out that $\mathbb{U}_2=\emptyset$ may hold, i.e., the shifted common residues assigned to $q$ are all in a form $\hat{r}^{+}_{c,l_1}$. For $\hat{r}^{\pm}_{c,l}=\widetilde{r}^{\pm}_{c,l}$, we have $|\widetilde{r}^{-}_{c,l_2}-\widetilde{r}^{+}_{c,l_1}|=|\widetilde{r}^{-}_{c,l_2}-(r^{+}_{c,l_1}+\Delta^{+}_{c,l_1})|<\frac{\Gamma}{2}$ based on (\ref{test1-real}). Since $|\Delta^{+}_{c,l_1}|<\frac{\Gamma}{4}$, $|\widetilde{r}^{-}_{c,l_2}-r^{+}_{c}|<\frac{3\Gamma}{4}$ holds. Thus, 
\begin{equation}
    \begin{aligned}
|\hat{X}-X|&=|q\Gamma+\frac{\sum_{l_1\in\mathbb{U}_1}\widetilde{r}^{+}_{c,l_1}+\sum_{l_2\in\mathbb{U}_2}\widetilde{r}^{-}_{c,l_2}}{L}-X|\\
&<|\lfloor\frac{X}{\Gamma}\rfloor\Gamma+r^{+}_{c}+\frac{3\Gamma}{4}-X|=\frac{3\Gamma}{4}.
    \end{aligned}
\end{equation}

When we apply operation 2 to case 2), $q$ will be one of the three kinds.
\begin{itemize}
\item [1] $q \equiv Y+1 \mod lcm( M_l \in \mathbb{U}_{11})$ and $q \equiv -Y \mod lcm( M_l \in \mathbb{U}_{21})$
\item [2] $q \equiv Y+1 \mod lcm( M_l \in \mathbb{U}_{11})$ and $q \equiv -Y-1 \mod lcm( M_l \in \mathbb{U}_{22})$ 
\item [3] $q \equiv Y \mod lcm( M_l \in \mathbb{U}_{12})$ and $q \equiv -Y-1 \mod lcm( M_l \in \mathbb{U}_{22})$
\end{itemize}
We start with the first situation, where $\Gamma >r^{+}_c+\Delta^{+}_{l_{11}}\geq\frac{\Gamma}{2}$ and $\Gamma> r^{-}_c+\Delta^{-}_{l_{21}}\geq\frac{\Gamma}{2}$. Likewise, under the assumption $lcm(M_l\in \mathbb{U}_{11})\geq lcm(M_l\in \mathbb{U}_{21})$, $q$ would be in a form $q=klcm(M_l\in \mathbb{U}_{11})+Y+1$. When $Y+1\not\equiv-Y\mod lcm(M_l\in \mathbb{U}_{21})$ and $\mathbb{U}_{21}\not=\emptyset$, $q>\omega$ with $k>0$, i.e., $q=Y+1$ is the only solution. If $Y+1\equiv-Y\mod lcm(M_l\in \mathbb{U}_{21})$ or $\mathbb{U}_{21}=\emptyset$, the shifted common residues assigned to $q$ are $\hat{r}^{+}_{c,l_{11}}$ and $\hat{r}^{-}_{c,l_{21}}$. For \[
\begin{aligned}
\hat{X} & =(Y+1)\Gamma+\frac{\sum_{l_{11}\in\mathbb{U}_{11}}\hat{r}^{+}_{c,l_{11}}+\sum_{l_{21}\in\mathbb{U}_{21}}\hat{r}^{-}_{c,l_{21}}}{L} \\
& =Y\Gamma+\frac{\sum_{l_{11}\in\mathbb{U}_{11}}\widetilde{r}^{+}_{c,l_{11}}+\sum_{l_{21}\in\mathbb{U}_{21}}\widetilde{r}^{-}_{c,l_{21}}}{L}
\end{aligned},\] we have $$|\hat{X}-X|=|Y\Gamma+\frac{\sum_{l_{11}\in\mathbb{U}_{11}}\widetilde{r}^{+}_{c,l_{11}}+\sum_{l_{21}\in\mathbb{U}_{21}}\widetilde{r}^{-}_{c,l_{21}}}{L}-X|<\frac{3\Gamma}{4}.$$
The proof is similar for the rest situations, which is omitted for simplicity. 
\end{proof}

\begin{cor}
\label{cor-real-one}
If $\min_{d=0,\pm 1}|r^{+}_l-r^{-}_l+dM_l\Gamma|>\Gamma$ holds for each $l$, $X$ can be recovered robustly error bounded by $\frac{\Gamma}{4}$ for any $X\in[0,D)$.
\end{cor}
\begin{proof}
Similarly, we only elaborate on case 2) for brevity.
$\min_{d=0,\pm 1} |r^+_l-r^-_l+dM_l\Gamma|>\Gamma$ denotes the closest distance between $r^+_l$ and $r^-_l$ on the circle of length $M_l\Gamma$.
\begin{equation}
\label{case2_residue_+}
 r^+_l=\langle X\rangle_{M_l\Gamma}=\langle Y\Gamma\rangle_{M_l\Gamma}+ r^+_c
\end{equation}
\begin{equation}
\label{case2_residue_-}
r^-_l=\langle - X\rangle_{M_l\Gamma}=\langle (-Y-1)\Gamma\rangle_{M_l\Gamma}+r^-_c
\end{equation}
For $\min| r^+_l- r^-_l+dM_l\Gamma|>\Gamma$, replacing $ r^+_l$ and $r^-_l$ with the right hand of (\ref{case2_residue_+}) and (\ref{case2_residue_-}) respectively, we obtain $$\min|\langle Y\Gamma\rangle_{M_l\Gamma}-\langle (-Y-1)\Gamma\rangle_{M_l\Gamma}+( r^+_c- r^-_c)+dM_l\Gamma|>\Gamma.$$
Since $0\leq r^+_c- r^-_c\leq\frac{\Gamma}{2}$, $$\min|\langle Y\Gamma\rangle_{M_l\Gamma}-\langle (-Y-1)\Gamma\rangle_{M_l\Gamma}+dM_l\Gamma|\neq 0,$$ i.e., for each $l$, $Y\not\equiv-Y-1\mod M_l$. Analogously, $Y+1\not\equiv-Y\mod M_l$ and $Y+1\not\equiv-Y-1\mod M_l$. When $q$ is in the form (\ref{op1-form}), since $Y\not\equiv-Y-1\mod M_l$, only if $\mathbb{U}_2=\emptyset$, will $q$ be within the range, i.e., the shifted common residues assigned to $q$ are in a form $\hat{r}^{+}_{c,l}$. Thus, $$|\hat{X}-X|=|Y\Gamma+\frac{\sum^{L}_{l=1}r^{+}_{c,l}}{L}-X|<\frac{\Gamma}{4}.$$
One can similarly verify the claim in other situations, which is omitted for simplicity.
\end{proof}
Hence, it is distinguishable if we classify residues of both $X$ and $-X$ into one clustering when $\min_{d =0,\pm 1} | r^+_l-r^-_l+dM_l\Gamma|>\Gamma$ holds for each $l$. 
If $r^{\pm}_l$ is uniformly distributed over the ring with a circumference of $M_l\Gamma$, 
$$\Pr(\min_{d=0,\pm 1} |r^+_l-r^-_l+dM_l\Gamma|>\Gamma)=\prod^{L}_{l=1}\frac{M_l-2}{M_l},$$  which can be close to $1$ given sufficiently large $M_l$.

\section{Proof of Theorem \ref{thm:real_fre}}
\label{pf:thm_real}
\begin{proof}
With a similar reasoning, the proof is organized into two parts. First, we prove that when $X_i\in[0,D)$, their folding numbers have a unique residue representation. Second, with proper residue classifications and sorting, one can uniquely derive a robust estimation for each $X_i$.

The proof for the first part is straightforward, where the residue representation of the folding number $\lfloor\frac{X_i}{\Gamma}\rfloor$ is unique. In the following, we denote $\lfloor\frac{X_i}{\Gamma}\rfloor$ by $Y_i$.
Clearly, since $X_i<D$, we have $$Y_i<\min_{\mathbb{U}\subset\{M_1,M_2,\dots,M_{\lceil\frac{L}{N}\rceil}\}}\frac{\prod_{l\in\mathbb{U}}M_l+\prod_{l\in\overline{\mathbb{U}}}M_l}{2}-1 = D_q,$$ where $\mathbb{U}\cup\overline{\mathbb{U}}=\{M_1,M_2,\dots,M_{\lceil\frac{L}{N}\rceil}\}$, satisfying (\ref{folding-number-range}).

Randomly selecting a residue from each $\mathcal{R}_l=\{\widetilde{r}^{+}_{i,l},\widetilde{r}^{-}_{i,l}|i=1,2,\dots,N\}$, we obtain a $L$-residue cluster $S_i=\{\widetilde{r}_{(i,l)}|i=1,2,\dots,N\}$, where $(i,l)$ denotes the index of the residues assigned to $S_i$. $S_i$ may contain residues from different pairs $\{X_i,-X_i\}$. Anyway, based on the pigeonhole principle, at least $\lceil \frac{L}{N}\rceil$ residues are from some $\{X_{i_0},-X_{i_0}\}$ in $S_i$. The estimated folding number obtained from $S_i$ is denoted by $q_i$. 

To avoid tedious discussion, we only consider the case where $r^{+}_{ic} \leq r^{-}_{ic}$.
Then, considering the tuple of residues $(\langle \lfloor \frac{ \widetilde{r}^+_{i,l}}{\Gamma} \rfloor \rangle_{M_l},\langle \lfloor \frac{ \widetilde{r}^-_{i,l}}{\Gamma} \rfloor \rangle_{M_l})$, it must be the residue of one of the four integer-pairs modulo $M_l$:
\begin{enumerate}
    \item $(Y_i, -Y_i-1)$, when $r^+_{ic,l} + \Delta^+_{i,l} \in [0,\Gamma),~ r^-_{ic,l} + \Delta^-_{i,l} \in [0,\Gamma)$;
    \item $(Y_i, -Y_i)$, when $r^+_{ic,l} + \Delta^+_{i,l} \in [0,\Gamma),~ r^-_{ic,l} + \Delta^-_{i,l} \in [\Gamma,2\Gamma)$;
    \item $(Y_i-1,-Y_i-1)$, when $r^+_{ic,l} + \Delta^+_{i,l} \in  (-\Gamma,0),~ r^-_{ic,l} + \Delta^-_{i,l} \in [0,\Gamma)$;
    \item $(Y_i-1,-Y_i )$, when $r^+_{ic,l} + \Delta^+_{i,l} \in (-\Gamma,0),~ r^-_{ic,l} + \Delta^-_{i,l} \in [\Gamma,2\Gamma)$.
\end{enumerate}

Perturbed by errors, $\langle\lfloor\frac{\widetilde{r}^{+}_{i,l}}{\Gamma}\rfloor\rangle_{M_l}$ can be the residue of either $\langle\lfloor \frac{X_i}{\Gamma} \rfloor\rangle_{M_l}$ or $\langle\lfloor \frac{X_i}{\Gamma} \rfloor-1\rangle_{M_l}$. To test that whether $\widetilde{r}_{(i,l)}$ in $S_i$ are all residues of one integer, similarly, we introduce a binary variable $\tau_{(i,l)}\in\{0,1\}$ to each $\widetilde{r}_{(i,l)}$. Then, the shifted common residues are $\hat{r}_{(ic,l)}=\langle\widetilde{r}_{(i,l)}\rangle_{\Gamma}-\tau_{(i,l)}\Gamma=\widetilde{r}_{(ic,l)}-\tau_{(i,l)}\Gamma$. Clearly, if all $\widetilde{r}_{(i,l)}$ are residues of $X_i$ assigned with proper $\tau_{(i,l)}$, $\hat{r}_{(ic,l)}$ would satisfy
\begin{equation}
\label{key-real-mul}
    |\hat{r}_{(ic,l_1)}-\hat{r}_{(ic,l_2)}|<\frac{\Gamma}{2}.
\end{equation}
Then, $\langle \lfloor \frac{ \widetilde{r}_{(i,l)}}{\Gamma}\rfloor + \tau_{(i,l)} \rangle_{M_l}$ must be the residue of one of the six integers modulo $M_l$: (a).$Y_i$, for $l \in \mathbb{U}_a$; (b).$Y_i-1$, for $l \in \mathbb{U}_b$; (c).$Y_i+1$, for $l \in \mathbb{U}_c$; (d).$-Y_i$, for $l \in \mathbb{U}_d$; (e).$-Y_i-1$, for $l \in \mathbb{U}_e$; (f).$-Y_i+1$, for $l \in \mathbb{U}_f$.

With a similar reasoning as previous proof, when $r^{+}_{ic} \in [\frac{\Gamma}{4}, \frac{\Gamma}{2}]$, only case 1) may happen in (\ref{key-real-mul}), which results in that only cases a), c), d) and e) can happen accordingly. 
It is noted that $|\mathbb{U}_a| \cdot |\mathbb{U}_c|=0$, i.e., one of $\mathbb{U}_a$ and $\mathbb{U}_c$ should be empty, since for any $l_1 \in \mathbb{U}_a$ and $l_2 \in \mathbb{U}_c$,
\begin{equation}
|\widetilde{r}_{ic,l_1}-(\widetilde{r}_{ic,l_2}-\Gamma)| \geq |(r^{+}_{ic} + \Delta^{+}_{i,l_1})-( r^{+}_{ic}+ \Delta^{+}_{i,l_2}-\Gamma)| > \frac{\Gamma}{2},
\end{equation}
violating (\ref{key-real-mul}).
Likewise, we have $|\mathbb{U}_d|\cdot|\mathbb{U}_e|=0$ and $|\mathbb{U}_c|\cdot|\mathbb{U}_e|=0$.
Thus, the recovered folding number $q_i$ will fall into one of the following forms in case 1): 
\begin{enumerate}
    \item $q_i \equiv Y_i \mod lcm(M_l, l\in \mathbb{U}_a)$ and $q_i \equiv -Y_i-1 \mod lcm(M_l, l\in \mathbb{U}_e)$
    \item $q_i \equiv Y_i+1 \mod lcm(M_l, l\in \mathbb{U}_c)$ and $q_i \equiv -Y_i \mod lcm(M_l, l\in \mathbb{U}_d)$.
    \item $q_i \equiv Y_i \mod lcm(M_l, l\in \mathbb{U}_a)$ and $q_i \equiv -Y_i \mod lcm(M_l, l\in \mathbb{U}_d)$.
\end{enumerate}

Without loss of generality, we assume that at least $\lceil \frac{L}{N}\rceil$ residues in $S_i$ are from $\{X_1,-X_1\}$, corresponding to the moduli set $\mathbb{U}_1=\{M_1,M_2,\dots,M_{l_k}\}$, where $l_k\geq\lceil\frac{L}{N}\rceil$.
Based on the above understanding, $q_i$ must be in one of the forms, subject to the qualification of (\ref{key-real-mul}).
\begin{itemize}
\item $q_i \equiv Y_1 \mod lcm(M_l, l\in \mathbb{U}_{1a})$\quad and\quad$q_i \equiv -Y_1-1 \mod lcm(M_l, l\in \mathbb{U}_{1e})$
\end{itemize}
\begin{itemize}
\item $q_i \equiv Y_1+1 \mod lcm(M_l, l\in \mathbb{U}_{1c})$\quad and\quad $q_i \equiv -Y_1 \mod lcm(M_l, l\in \mathbb{U}_{1d})$
\end{itemize}
\begin{itemize}
\item $q_i \equiv Y_1 \mod lcm(M_l, l\in \mathbb{U}_{1a})$\quad and\quad $q_i \equiv -Y_1 \mod lcm(M_l, l\in \mathbb{U}_{1d})$
\end{itemize}
We study the first situation first, where $q_i$ satisfies
\begin{equation}
\label{folding-number-real}
   \left \{
  \begin{array}{l}
     q_i\equiv Y_1\mod lcm(M_l\in \mathbb{U}_{1a})\\
   q_i\equiv -Y_1-1\mod lcm(M_l\in \mathbb{U}_{1e})\\
  q_i\equiv \lfloor\frac{\widetilde{r}_{(i,l')}}{\Gamma}\rfloor+\tau_{(i,l')}\mod M_{l'}
  \end{array}.
  \right. 
\end{equation}
where $\mathbb{U}_{1a}\cup\mathbb{U}_{1e}=\mathbb{U}_1$ and $l'\in\{l_k+1,l_k+2,\dots,L\}$.
Without loss of generality, we assume that $lcm(M_l\in\mathbb{U}_{1a})>lcm(M_l\in\mathbb{U}_{1e})$.
Then, $q_i=klcm(M_l\in\mathbb{U}_{1a})+Y_1$.
Clearly, when $Y_1\not\equiv -Y_1-1\mod lcm(M_l\in\mathbb{U}_{1e})$, $q_i$ will contradict the dynamic range assumed with $k\geq 1$.
Only when $Y_1\equiv -Y_1-1\mod lcm(M_l\in\mathbb{U}_{1e})$ or $lcm(M_l\in\mathbb{U}_{1e})=\emptyset$ is satisfied will $q_i\in[0,D_q)$, i.e., $q_i=Y_1$.
So (\ref{folding-number-real}) changes to
\begin{equation}
\label{folding-number-real-1}
   \left \{
  \begin{array}{l}
  q_i\equiv Y_1\mod lcm(M_l\in \mathbb{U}_{1})\\
  q_i\equiv \lfloor\frac{\widetilde{r}_{(i,l')}}{\Gamma}\rfloor+\tau_{(i,l')}\mod M_{l'}
  \end{array}
  \right. 
\end{equation}
Thus, we have $q_i=klcm(M_l\in \mathbb{U}_{1})+Y_1$, where $lcm(M_l\in \mathbb{U}_1)> D_q$.
Similarly, if at least one $l'\in[l_k+1,L]$ satisfies $Y_1 \not\equiv \lfloor\frac{\widetilde{r}_{(i,l')}}{\Gamma}\rfloor+\tau_{(i,l')}\mod M_l'$, $q_i$ will exceed the range $D_q$, leading to a contradiction.

That is to say, $q_i= Y_1$ is the only efficient solution with at least $\lceil\frac{L}{N}\rceil$ residues are from $\{X_1,-X_1\}$ and $lcm(M_l\in \mathbb{U}_{1a})> lcm(M_l\in \mathbb{U}_{1e})$.
The shifted common residues assigned to $q_i$ are $\hat r^{+}_{1c,l}$ and $\hat r^{\pm}_{ic,l}$. Based on (\ref{key-real-mul}), we have $|\hat r^{\pm}_{ic,l_1}-\hat r^{+}_{1c,l_2}|=|\hat r^{\pm}_{ic,l_1}-(r^{+}_{1c}+\Delta^{+}_{1c,l_2})|<\frac{\Gamma}{2}$, i.e., $|\hat r^{\pm}_{ic,l_1}-r^{+}_{1c}|<\frac{3\Gamma}{4}$.
Therefore, $|\hat{X}_1-X_1|$ is equal to
\begin{equation}
|Y_1\Gamma+\frac{\sum_{l\in\mathbb{U}_{1a}}\hat{r}^{+}_{1c,l}+\sum_{l\in\mathbb{U}_{1e}}\hat{r}^{-}_{1c,l}+\sum_{l\in\overline{\mathbb{U}}_1}\hat{r}^{\pm}_{ic,l}}{L}-X_1|<\frac{3\Gamma}{4}.
\end{equation}

In the following, we study the second situation, where $q_i \equiv Y_1+1 \mod lcm(M_l, l\in \mathbb{U}_{1c})$\quad and\quad $q_i \equiv -Y_1 \mod lcm(M_l, l\in \mathbb{U}_{1d})$.
Thus, (\ref{folding-number-real}) becomes  
\begin{equation}
\label{folding-number-real-2}
   \left \{
  \begin{array}{l}
  q_i\equiv Y_1+1\mod lcm(M_l\in \mathbb{U}_{1c})\\
  q_i\equiv -Y_1\mod lcm(M_l\in \mathbb{U}_{1d})\\
  q_i\equiv \lfloor\frac{\widetilde{r}_{(i,l')}}{\Gamma}\rfloor+\tau_{(i,l')}\mod M_{l'}
  \end{array}
  \right. 
\end{equation}
Similarly, under the assumption $lcm(M_l\in\mathbb{U}_{1c})>lcm(M_l\in\mathbb{U}_{1d})$, $q_i=Y_1+1$ is the only efficient solution within $D_q$. So $|\hat{X}_1-X_1|$ is equal to
\begin{equation}
\begin{aligned}
&|(Y_1+1)\Gamma+\frac{\sum_{l\in\mathbb{U}_{1c}}\hat{r}^{+}_{1c,l}+\sum_{l\in\mathbb{U}_{1d}}\hat{r}^{-}_{1c,l}+\sum_{l\in\overline{\mathbb{U}}_{1}}\hat{r}^{\pm}_{ic,l}}{L}-X_1|\\
&<|\lfloor\frac{X_1}{\Gamma}\rfloor\Gamma+\frac{\sum_{l\in\mathbb{U}}({r}^{+}_{1c}+\frac{3\Gamma}{4})}{L}-X_1|=\frac{3\Gamma}{4}.
\end{aligned}
\end{equation}
We can reach the same conclusion $|\hat{X_i}-X_i|<\frac{3\Gamma}{4}$ in the last situation, which is omitted. Q.E.D.
\end{proof}
\begin{algorithm}
\caption{Robust Remaindering Decoding of Multiple Frequencies in Real Waveform}
\textbf{Input}: Moduli set: $\mathcal{M}=\{m_1=M_l\Gamma|l=1,2,\dots,L\}$, where $M_l$ are sorted in ascending order;\\
Residue Sets: $\mathcal{R}_l=\{\widetilde{r}^{\pm}_{i,l}|i=1,2,\dots,N\}$, $l=1,2,\dots,L$.
\begin{algorithmic}[1]

\STATE \textbf{Repeat}: Propose a clustering assignment
\STATE Following the proposed clustering by selecting one residue from each $\mathcal{R}_l$ to obtain a $L$-residue clustering $S_i=\{\widetilde{r}_{(i,l)}|l=1,2,\dots,L\}$, $i=1,2,\dots,N$, where $(i,l)$ denotes the index of the residues assigned to $S_i$.
\STATE Assign a parameter $\tau_{(i,l)}\in\{0,1\}$ to each residue in $S_i$ randomly.
\STATE Calculate the shifted common residues $\hat{r}_{(ic,l)}=\langle\widetilde{r}_{(i,l)}\rangle_{\Gamma}-\tau_{(i,l)}\Gamma$ corresponding to each $\widetilde{r}_{(i,l)}$ in $S_i$.
\STATE For each $S_{i}$, calculate $q_i\equiv\frac{\widetilde{r}_{(i,l)}-\hat{r}_{(ic,l)}}{\Gamma}\mod M_l$ via CRT.
\STATE\textbf{Until}: There exist $N$ clusterings $S_i$ such that $q_i\in[0,D_q)$ and each residue in $S_i$ satisfies (\ref{key-real-mul}) for $i=1,2,\dots,N$, where $D_q=\min_{\mathbb{U}}\frac{lcm(M_l\in \mathbb{U})+lcm(M_l\in(\overline{\mathbb{U}})}{2}-1$ and $\mathbb{U}\cup\overline{\mathbb{U}}=\{M_1,M_2,\dots,M_{\lceil\frac{L}{N}\rceil}\}$.
\end{algorithmic}
\textbf{Output}: $\hat{X_i}=q_i\Gamma+\frac{\sum^{L}_{l=1}\hat{r}_{(ic,l)}}{L}$.
\label{alg:real_decoding}
\end{algorithm}

\end{document}